\documentclass[11pt]{article}
\usepackage{inputenc,amssymb,amsthm,enumerate,mathtools,mathrsfs}
\usepackage{amsmath}

\usepackage[all]{xy}
\usepackage{pb-diagram}
\usepackage{tikz}
\usepackage{graphicx,xcolor}
\usepackage{xspace}
\usetikzlibrary{matrix,decorations.pathreplacing,calligraphy}

\newtheorem{thm}{Theorem}

\newtheorem{prop}[thm]{Proposition}
\newtheorem{lemma}[thm]{Lemma}
\theoremstyle{definition}
  \newtheorem{defi}[thm]{Definition}
  \newtheorem{ex}[thm]{Example}
  \newtheorem{monad}{Monad}
\theoremstyle{remark}
  
  \newtheorem*{rem}{Remark}

\DeclareMathOperator{\set}{\mathtt{Set}}

\DeclareMathOperator{\Alg}{\mathtt{Alg}}

\DeclareMathOperator{\rec}{\mathsf{Rec}}

\DeclareMathOperator{\mso}{\mathsf{MSO}}

\newcommand{\Bag}{\mathcal{B}}
\newcommand{\StB}{{\mathsf{B}}}
\newcommand{\JL}{\lambda}
\newcommand{\StuB}{{\underline{\mathsf{B}}}}
\newcommand{\StAA}{{\mathsf{A\mkern-5muA}}}
\newcommand{\StAB}{{\mathsf{A\mkern-5muB}}}
\newcommand{\StBA}{{\mathsf{B\mkern-5muA}}}
\newcommand{\StBB}{{\mathsf{B\mkern-5muB}}}

\newcommand{\FL}{\mathbf{FL}}

\newcommand{\Sta}{\mathsf{a}}
\newcommand{\Stb}{\mathsf{b}}
\newcommand{\Stc}{\mathsf{c}}
\newcommand{\StL}{\mathsf{L}}
\newcommand{\StR}{\mathsf{R}}

\newcommand{\PP}{\mathcal{P}^+}
\renewcommand{\P}{\mathcal{P}}
\newcommand{\id}{\mathrm{id}}

\newcommand\mydot{\mathbin{\vcenter{\hbox{\scalebox{.5}{$\bullet$}}}}}

\newcommand{\msotext}{{\sf MSO}\xspace}
\newcommand{\st}[1]{\{#1\}}
\title{Monadic Monadic Second Order Logic}
\author{Miko\l{}aj Boja\'nczyk\thanks{Research supported by the European Research Council (ERC) under the European Union's Horizon 2020 research and innovation programme (ERC consolidator grant LIPA, agreement no. 683080).} \and Bartek Klin$^*$ \and Julian Salamanca$^*$}
\date{\em University of Warsaw}

\begin{document}
\maketitle

\begin{abstract}
	One of the main reasons for the correspondence of  regular languages and monadic  second-order logic  is that the class of regular languages is closed under images of surjective letter-to-letter homomorphisms. This closure property holds for structures such as finite words, finite trees, infinite words, infinite trees, elements of the free group, etc. Such structures can be modelled using monads.  In this paper, we study which structures (understood via monads in the category of sets) are such that the class of regular languages (i.e.~languages recognized by finite algebras) are closed under direct images of surjective letter-to-letter homomorphisms.
	 We provide diverse sufficient conditions for a monad to satisfy this property. We also present numerous examples of monads, including positive examples that do not satisfy our sufficient conditions,  and  counterexamples where the closure property fails. 
\end{abstract}

\pagestyle{myheadings}

\vspace{12pt}
\setlength{\parindent}{0.8cm}

\section{Introduction}
	A seminal result in automata theory is that recognizable  languages (i.e.~languages recognized by finite state devices)  are exactly those that can be defined in monadic second-order logic (\msotext). This result was originally shown for finite words
	 by (independently) B\"uchi~\cite[Corollary 4]{buchi_weak_1990}, Elgot~\cite[Theorem 5.3]{elgot_decision_1961} and Trakhtenbrot~\cite{trakhtenbrot1961finite}, but it is also true for:
\begin{itemize}
	\item $\omega$-words, shown by B\"uchi~\cite[Theorem 1]{buchi_decision_1990};
	\item finite trees,  shown by Doner~\cite[Corollary 3.8]{doner_tree_1970} and Thatcher and Wright~\cite[Theorem 14]{thatcher_generalized_1968};
	\item infinite trees,  shown by Rabin~\cite[Theorem 1.7]{rabin_decidability_1969};
	\item countable linear orders, shown by Carton, Colcombet and Puppis~\cite[Theorem 3]{carton_regular_2011} building on    Shelah~\cite[Theorem 6.2]{shelah_monadic_1975};
	\item graphs of bounded treewidth, which follows from Courcelle's Theorem~\cite[Theorem 4.4]{courcelle_monadic_1990} and its converse~\cite[Theorem 2.10]{bojanczyk_definability_2016}.
\end{itemize}
For more about these results, see the survey by Thomas~\cite{thomas}.

The large number of examples calls for a more systematic framework, where the notion of a composable structure (be it a finite word, $\omega$-word, finite tree etc.) would be a parameter, and a characterization of the expressive power of \msotext, a result. In our view, the lack of such a framework is a manifestation of what Samson Abramsky has recently called
\begin{quote}\small
a remarkable divide in the field of logic in Computer Science, between two distinct strands: one focussing on semantics and compositionality (``Structure''), the other on expressiveness and complexity (``Power''). It is remarkable because these two fundamental aspects of our field are studied using almost disjoint technical languages and methods, by almost disjoint research communities. \cite[p.~1]{abramsky18}.
\end{quote}

A generic approach to \msotext was proposed in~\cite{boj1}, using monads. Monads are a standard categorical tool to study compositionality, and they are firmly rooted in the ``Structure'' strand of Theoretical Computer Science. Applying them to study the expressive power of \msotext is a step towards building a bridge over the divide.

Monads are general enough to capture structures such as words, trees, graphs, etc.~but specific enough to describe concepts such as recognizable languages (as observed by Eilenberg and Wright in~\cite[Section 11]{EILENBERG1967452}), syntactic algebras~\cite[Section 3]{boj1} or pseudovarieties~\cite[Section 4]{boj1}. Another advantage of monads is that they can model infinite objects (such as infinite words or trees, which are central topics in automata theory), which is not the case for some alternative approaches, such as Steinby's  approach via universal algebra~\cite{steinby_general_1998}.

When discussing the relationship of \msotext and recognizable languages,
there are two implications to consider.

\begin{itemize}
	\item \emph{Recognizable $\Rightarrow$ definable in \msotext.} This implication is easy  for structures such as finite words or trees, where \msotext can be used to define some canonical decomposition. In other cases, the proof can be much harder.  Indeed, the proofs are relatively recent in the cases of countable linear orders~\cite[Theorem 3]{carton_regular_2011}, graphs of bounded treewidth~\cite[Theorem 2.10]{bojanczyk_definability_2016}, or graphs of bounded linear cliquewidth~\cite[Theorem 3.5]{bojanczyk_definable_2018}. The implication is known to fail for graphs of unbounded treewidth or cliquewidth~\cite[Proposition 4.36]{courcelle_graph_2012}, and it also fails for infinite trees under a naive definition 
of finite algebras~\cite[Section 4]{bojanczyk_non-regular_2018}. Sometimes, e.g.~for graphs of bounded cliquewidth, the implication remains an open question~\cite[Section 8]{bojanczyk_definable_2018}. A general understanding of this implication seems to be a hard problem, and we do not make any attempts in that direction in this paper.

	\item \emph{Definable in \msotext $\Rightarrow$ recognizable.} In all known cases, when recognizable languages are defined in terms of finite algebras,  this implication is relatively straightforward.\footnote{For this it is important that algebras, and not automata, are used as the notion of recognizability. For example, the hard part of Rabin's theorem is showing that languages recognized by nondeterministic automata are closed under complementation, see~\cite[Theorem 1.5]{rabin_decidability_1969} or~\cite[Theorem 6.2]{thomas}. This difficulty disappears when considering algebras, where complementation is achieved simply by flipping the accepting set. 
	} Roughly speaking,  languages recognized by finite algebras  are automatically closed under Boolean combinations, and closure under existential quantification is proved using  some kind of powerset construction.
This argument is so deceptively simple that~\cite[Lemma 6.2]{boj1} wrongly claimed that it works for every monad. One purpose of our paper is to correct this mistake, show counterexamples to the claim, and study conditions that make the implication hold; the resulting landscape of monads turns out to be quite interesting.
\end{itemize}

\paragraph*{Motivating example:  finite words}
To explain how \msotext can be defined for an arbitrary monad, we begin by describing \msotext for finite words in a manner which can be translated to a more generic setting.

A word $w$ over an alphabet $\Sigma$ can be understood as a relational structure~\cite[Section 2.1]{thomas}, denoted by $\underline w$,  where the universe is the positions of the word, there is a binary predicate $x < y$ for the order on positions, and for every $a \in \Sigma$ there is a unary predicate $Q_a$ which selects positions that have label $a$. To define properties of a word $w$, we can use sentences of first-order logic or monadic second-order logic 
over the vocabulary of $\underline w$. For example, if $w \in \st{a,b}^*$ then the sentence
\begin{align*}
\underbrace{\exists x\big[Q_a(x)}_{\substack{ \text{there is}\\ \text{a position} \\\text{with label $a$}}}\ \land \ \underbrace{\forall y \ (y > x \Rightarrow Q_a(x))\big]}_{\substack{\text{such that all later positions}\\ \text{also have label $a$}}}
\end{align*}
is true in $\underline w$ if and only if $w$ belongs to the regular language $(a+b)^*a^+$. This sentence uses only first-order logic (quantification over positions), but some regular languages need monadic second-order logic (quantification over sets of positions). A famous example is the regular language of words of even length, which  is defined by the sentence
\begin{align*}
\underbrace{\exists X}_{\substack{\text{there is}\\ \text{a set of}\\ \text{positions}}} \bigwedge \ \begin{cases}
	\underbrace{\forall x \ \exists y \ y \le x \land y \in X}_{\text{the first position is in $X$}}\\
	\underbrace{\forall x \ \exists y \ y \ge x \land y \not \in X}_{\text{the last position is not in $X$}}\\
	\underbrace{\forall x \ \forall y \ (x < y \land \neg (\exists z\ x < z < y)) \Rightarrow (x \in X \Leftrightarrow y \not \in X)}_{\text{for every two consecutive positions, exactly one is in $X$}}.
\end{cases}
\end{align*}
The theorem of B\"uchi, Elgot and Trakhtenbrot that we mentioned above says that a language of finite words is regular if and only if it can be defined by a sentence of \msotext. The ``regular $\Rightarrow$ definable in $\msotext$'' implication
can be proved by using the logic to formalise the semantics of a nondeterministic finite automaton (or a regular expression). For the converse implication, one shows that the class of regular languages has all the closure properties that are used in \msotext. This observation is formalised in the following result. (A letter-to-letter homomorphism is defined to be a function of the form $f^* : \Sigma^* \to \Gamma^*$ for some $f : \Sigma \to \Gamma$.)
\begin{prop}\label{propclosuremso}[cf. \cite[Lemma 6.1]{boj1}]\label{prop:boj1}
	For finite words, the class of \msotext definable languages is the least class of languages over finite alphabets that contains the languages $0^*\subseteq\{0,1\}^*$ and $0^*1^*\subseteq\{0,1\}^*$ and is closed under union, intersection, complement, inverse images and direct images along surjective letter-to-letter homomorphisms.
\end{prop}
\begin{proof}(Sketch)  Following~\cite[Section 2.3]{thomas}, one can eliminate first-order variables from \msotext, and keep only the monadic second-order variables. Instead of the usual atomic predicates $x < y$ and $Q_a(x)$, which use first-order variables, one uses atomic predicates
	\begin{align*}
	\underbrace{X \subseteq Y}_{\text{set inclusion}} \qquad \underbrace{X < Y}_{\substack{\text{every position in $X$ is}\\ \text{before every position in $Y$} }} \qquad  \underbrace{X \subseteq Q_a}_{\text{all positions in $X$ have label $a$}}
	\end{align*}
	A formula with free variables  $\phi(X_1,\ldots,X_n)$ over an alphabet $\Sigma$ can be seen as a language  over an extended alphabet  $\Sigma\times 2^n$. The idea is to interpret a word over the extended alphabet as a word together with a valuation of the sets $X_1,\ldots,X_n$, as indicated by the $n$ bits stored by every position.
For instance,
if $\phi(X,Y)=X\subseteq Y$ and $\psi(X,Y)=Y\subseteq X$, then the word
\[
w=(a,0,0)(a,1,0)(b,1,1),
\]
does not satisfy $\phi$, but satisfies $\psi$.

The set of words that satisfy $\exists X\ \phi(X,Y_1,\ldots,Y_n)$ is exactly the direct image of the set of words
that satisfy $\phi(X,Y_1,\ldots,Y_n)$ under the letter-to-letter homomorphism $\pi^*$ which is obtained by lifting to words  the projection
\begin{align*}
	\pi:\Sigma\times 2^{n+1}\to \Sigma\times 2^n \qquad  (a,x,y_1,\ldots,y_n) \mapsto (a,y_1,\ldots,y_n)
\end{align*}
In this sense, the second-order existential quantifier is abstractly captured by direct images under (surjective) letter-to-letter homomorphims.

Given the above observations, one can show the theorem by induction of the structure of \msotext formulas. The language $0^*\subseteq\{0,1\}^*$ corresponds to the predicates $X\subseteq Y$ and $X \subseteq Q_a$ by means of an inverse image (the idea is that the letter $0$ represents the positions which satisfy $X \subseteq Y$ and $1$ represents the other positions, likewise for $X \subseteq Q_a$). The language $0^*1^*\subseteq\{0,1\}^*$ corresponds to the predicate $X<Y$, again by an inverse image.
	Direct images of surjective letter-to-letter homomorphisms correspond to the quantifier $\exists$. Closure under union, intersection and complement come from the logical
	connectives $\lor$, $\land$ and $\neg$, respectively.
\end{proof}

Motivated by Proposition~\ref{propclosuremso}, one can define an abstract version of \msotext in any monad, at least over the category of sets, as the least class of languages with the closure properties stated in the proposition. This definition was proposed in~\cite[Section 6.1]{boj1}, and is described in more detail in Section~\ref{sec:monadic:mso}. A parameter of the definition is a choice of atomic languages (such as $0^*$ and $0^*1^*$ in Proposition~\ref{propclosuremso}). If the atomic languages are recognizable, and the class of recognizable languages 
has the required closure properties, then all \msotext definable languages are going to be recognizable.  As mentioned previously,  most closure properties for recognizable languages are straightforward: Boolean combinations and inverse images under surjective letter-to-letter homomorphisms (or, indeed, under arbitrary homomorphisms). There is  one exception: direct images under surjective letter-to-letter homomorphisms. The main topic of this paper is understanding this last closure property.

As we will see, there are examples of monads for which the closure property fails. Also, when the property holds, it can hold for various reasons. We identify  three rather diverse sufficient conditions:
\begin{enumerate}
	\item Some monads admit a powerset construction for algebras, which entails closure under direct images. In Section~\ref{sec:suffcond-jacobs}, we identify a sufficient condition on the monad (we call such monads weakly epi-cartesian) which guarantees that a powerset construction works.   Many monads are weakly epi-cartesian, including  monads for all structures that have been traditionally considered in automata theory (words, trees, etc.).

	As a by-product, we prove that weakly cartesian monads admit distributive laws over the powerset monad (Theorem~\ref{thm:distlaw}). This result is of independent interest: it extends previous results by Jacobs~\cite{jacobsTSC}, and it contrasts with recent negative results of Klin and Salamanca~\cite{KliSal18} and Zwart and Marsden~\cite{zwart19}.
	\item Another sufficient condition is having  a Mal'cev term (see Section~\ref{sec:malcev}). For example, the free group monad has a Mal'cev term,
but it does not admit a powerset construction.
	\item Yet another reason is that the monad preserves finiteness, which implies that all languages over finite alphabets are trivially recognizable.
	Although seemingly trivial, this condition can cover interesting examples, such as the monad of idempotent monoids, where preservation of finiteness is true but not obvious.
\end{enumerate}
We will also show several examples of monads where the closure property holds, but which do not fall into any of the three classes described above. The project of understanding \msotext for monads -- even in the restricted setting of the category of sets -- is still far from complete. We hope, however, that the reader will be intrigued by the rich variety of examples -- we show~\ref{mon:lastmonad} of them altogether.

\paragraph*{Structure of the paper.}
In Section~\ref{sec:monads} we recall basic notions and results about monads, followed by an abstract monadic definition of \msotext{} in Section~\ref{sec:monadic:mso}. In Sections~\ref{sec:suffcond} and~\ref{sec:suffcond-jacobs} we present three sufficient conditions which guarantee that recognizable languages are closed under direct images of surjective letter-to-letter homomorphisms. We show numerous examples there. Other positive examples, which do not satisfy any of the three conditions, are shown in Section~\ref{sec:sporadic}. Section~\ref{sec:counterexamples} lists a few counterexamples where the closure property fails. In Section~\ref{sec:related} we illustrate what would happen if somewhat stronger closure properties were required of a monad. We conclude in Section~\ref{sec:futurework} by sketching main directions of future work.

\section{Monads}\label{sec:monads}

We now introduce some basic concepts and intuitions related to monads. Everything in this section is completely standard (see e.g.~\cite{AwodeyCT,maclaneCT}).

One of the several possible intuitive understandings of monads is that they formalize ``ways to collect things''. Given a set $\Sigma$ (understood as an alphabet), a monad returns a set $T\Sigma$ of collections, or structures, built out of letters from $\Sigma$. The kind and shape of these structures depends on the monad: for example, the finite list monad sets $T\Sigma$ to be the set of all finite sequences of elements of $\Sigma$, and the powerset monad returns the set of all subsets of $\Sigma$. A monad must provide further structure:
\begin{itemize}
\item A way to apply functions to structures element-wise; formally, a function $f:\Sigma\to \Gamma$ should yield a function $Tf:T\Sigma\to T\Gamma$ so that identity functions and function composition are preserved (formally, this makes $T$ a functor);
\item A way to build ``singleton'' structures, formalized as a function $\eta_{\Sigma}:\Sigma\to T\Sigma$, called the {\em unit}, for each alphabet $\Sigma$;
\item A way to ``flatten'' structures of structures to single-layer structures, formalized as a function $\mu_\Sigma:T(T\Sigma)\to T\Sigma$, called the {\em multiplication}, for each alphabet $\Sigma$.
\end{itemize}
These ingredients are subject to a few axioms. First, both the unit and the multiplication must be {\em natural}, i.e., they must be invariant under arbitrary renamings of elements. Formally, the ``naturality squares''
\begin{equation}\label{eq:naturality}
\vcenter{\xymatrix{
	\Sigma\ar[r]^{\eta_\Sigma}\ar[d]_f & T\Sigma\ar[d]^{Tf} \\
	\Gamma\ar[r]_{\eta_\Gamma} & T\Gamma
}}\qquad\qquad
\vcenter{\xymatrix{
	TT\Sigma\ar[r]^{\mu_\Sigma}\ar[d]_{TTf} & T\Sigma\ar[d]^{Tf} \\
	TT\Gamma\ar[r]_{\mu_\Gamma} & T\Gamma
}}
\end{equation}
must commute for every function $f:\Sigma\to \Gamma$. Furthermore, the multiplication operation must be associative and the unit must be the actual two-sided unit for it, in the sense made formal in the following definition:

\begin{defi}
A {\em monad} $(T,\eta,\mu)$ (on the category $\set$ of sets and functions) is a functor $T:\set\to\set$ together with natural transformations $\eta:Id\Rightarrow T$ and $\mu:TT\Rightarrow T$, subject to axioms:
\begin{equation}\label{eq:monadlaws}
\vcenter{\xymatrix{
T\Sigma\ar[r]^-{\eta_{T\Sigma}}\ar@{=}[rd] & TT\Sigma\ar[d]_{\mu_\Sigma} & T\Sigma\ar@{=}[ld]\ar[l]_-{T\eta_\Sigma} \\
& T\Sigma}} \qquad\qquad
\vcenter{\xymatrix{
TTT\Sigma\ar[r]^{T\mu_\Sigma}\ar[d]_{\mu_{T\Sigma}} & TT\Sigma\ar[d]^{\mu_\Sigma} \\
TT\Sigma\ar[r]_{\mu_\Sigma} & T\Sigma
}}
\end{equation}
for all sets $\Sigma$.
\end{defi}

We will usually denote a monad $(T,\eta,\mu)$ simply by $T$. If a risk of confusion arises (i.e. in the presence of more than one monad), its unit and multiplication will then be called $\eta^T$ and $\mu^T$.

Here are a few standard examples of monads.

\begin{monad}\label{mon:list}
The {\em list monad}, also known as the {\em free monoid monad}, is defined by $T\Sigma=\Sigma^*$, the set of finite words over $\Sigma$, with $Tf(x_1\cdots x_n) = f(x_1)\cdots f(x_n)$ for $f:\Sigma\to \Gamma$. The unit is defined by singleton words: $\eta_\Sigma(x)=x$, and multiplication by words concatenation: $\mu_\Sigma((w_1)(w_2)\cdots (w_n))=w_1w_2\cdots w_n$.
\end{monad}

\begin{monad}\label{mon:semigroup}
The same definitions restricted to non-empty lists form the {\em non-empty list monad} $T\Sigma=\Sigma^+$, also called the {\em free semigroup monad}.
\end{monad}

\begin{monad}\label{mon:power}
The {\em powerset monad} is defined by $T\Sigma={\cal P} \Sigma$, the set of all subsets of $\Sigma$, with the action on functions defined by direct image: ${\cal P} f(\Delta)=\overrightarrow{f}(\Delta)$ for $f:\Sigma\to \Gamma$ and $\Delta\subseteq \Sigma$. The unit is defined by singletons: $\eta_\Sigma(x)=\{x\}$, and multiplication by set union: $\mu_\Sigma(\Phi) = \bigcup\Phi$ for $\Phi\subseteq{\cal P}\Sigma$.
\end{monad}

\begin{monad}\label{mon:npower}
The {\em nonempty powerset monad} $\PP$ is defined exactly the same but with $T\Sigma$ restricted to nonempty subsets of $\Sigma$.
\end{monad}

\begin{monad}\label{mon:finpower}
The {\em finite powerset monad} is also defined the same but with $T\Sigma$ restricted to finite subsets of $\Sigma$. This is also called the {\em free semilattice monad}.
\end{monad}

\begin{monad}\label{mon:bag}
The {\em bag monad} $\Bag$, also known as the {\em multiset monad} or the {\em free commutative monoid monad}, is defined so that $\Bag \Sigma$ is the set of functions from $\Sigma$ to $\mathbb{N}$ that are zero almost everywhere, and the action on functions is defined by
\[
	\Bag f(\beta)(y) = \sum_{x\in f^{-1}(y)}\beta(x);
\]
the sum is well-defined since $\beta$ returns zero almost everywhere. The unit of $\Bag$ maps an element $x\in \Sigma$ to the function $[x\mapsto 1]$ that maps $x$ to $1$ and is zero everywhere else. Multiplication is defined by:
\[
	\mu_\Sigma(\beta)(x) = \sum_{\phi\in \Bag \Sigma}\beta(\phi)\times\phi(x).
\]
\end{monad}

\medskip

We shall be looking at numerous examples of monads in the following. Many of them are best presented in terms of operations and equations. A useful recipe for defining a (finitary) monad on $\set$ begins by considering an algebraic {\em signature}, i.e., a set of operation symbols, each with an associated finite arity. With a signature fixed, the notion of a {\em term} over a set $X$ of variables is defined as usual: a variable $x\in X$ is a term, and $\sigma(t_1,\ldots,t_n)$ is a term whenever $\sigma$ is a symbol of arity $n$ and $t_1,\ldots,t_n$ are terms. An {\em equation} is a pair of terms (over the same set of variables). A set of equations defines, for any alphabet $\Sigma$, a congruence relation on the set of terms over $\Sigma$ in the expected way: it is the least equivalence relation that is compatible with the operations and satisfies all the equations. One then defines $T\Sigma$ to be the set of equivalence classes of terms over $\Sigma$, under that congruence relation. For the unit, $\eta_\Sigma(x)$ is (the equivalence class of) the variable $x$, and multiplication is defined by term substitution. It is standard to check that these ingredients form a monad on $\set$, and the original set of equations is then called an {\em equational presentation} of that monad.

\begin{ex}
Let the signature contain a binary symbol $\mydot$ and a constant (i.e., a symbol of arity zero) $e$. The unit and associativity equations:
\[
	e\mydot x = x \mydot e = x \qquad\qquad (x\mydot y)\mydot z = x\mydot (y\mydot z)
\]
form a presentation of the free monoid monad (Monad~\ref{mon:list}). Adding a further equation for commutativity:
\[
	x\mydot y = y\mydot x
\]
one obtains the free commutative monoid monad (Monad~\ref{mon:bag}). Adding yet another equation:
\[
	x\mydot x = x
\]
yields a presentation of the finite powerset monad (Monad~\ref{mon:finpower}).
\end{ex}

\begin{defi}\label{def:eilenberg-moore}
Given a monad $T$, an (Eilenberg-Moore) {\em algebra} for $T$ (shortly, a {\em $T$-algebra}) is a set $A$ together with a function $\alpha: TA\to A$ subject to two axioms:
\[\xymatrix{
A\ar[r]^{\eta_A}\ar@{=}[rd] & TA\ar[d]^\alpha \\
& A
}\qquad\qquad
\xymatrix{
TTA\ar[r]^{\mu_A}\ar[d]_{T\alpha} & TA\ar[d]^{\alpha} \\
TA\ar[r]_{\alpha} & A.
}\]
A $T$-algebra {\em homomorphism} from $\mathbf{A}=(A,\alpha)$ to $\mathbf{B}=(B,\beta)$ is a function $h:A\to B$ such that the diagram
\[\xymatrix{
TA\ar[r]^{Th}\ar[d]_{\alpha} & TB\ar[d]^{\beta} \\
A\ar[r]_h & B
}\]
commutes. $T$-algebras and their homomorphisms form a category, denoted $\Alg(T)$.
\end{defi}

It is easy to check that for any set $\Sigma$, the set $T\Sigma$ with $\mu_\Sigma:TT\Sigma\to T\Sigma$ is a $T$-algebra. This is the {\em free $T$-algebra over $\Sigma$}; its fundamental property is that for any $T$-algebra $\mathbf{A}$, homomorphisms from $T\Sigma$ to $\mathbf{A}$ are in bijective correspondence with functions from $\Sigma$ to $A$.

\begin{ex}
For $T$ the list monad (Monad~\ref{mon:list}), a $T$-algebra on $A$ is a function that interprets arbitrary finite sequences over $A$ as elements of $A$. The two axioms of $T$-algebras mean that the function is associative in the obvious sense; indeed, it easy to check that $T$-algebras correspond to monoids. Moreover, $T$-algebra homomorphisms are monoid homomorphisms. $T\Sigma=\Sigma^*$ is the free monoid over $\Sigma$.

The same schema applies to any equationally definable class of algebras. We can therefore consider ``free X monad'' where X can stand for semigroup, commutative semigroup, group, abelian group, lattice, distributive lattice, Boolean algebra and so on, with algebras for the ``free X monad'' being exactly X's.
\end{ex}

\section{Monadic Monadic Second Order Logic}\label{sec:monadic:mso}

We now define our abstract monadic $\mso$, inspired by Proposition \ref{propclosuremso}. First we define a general notion of a recognizable language \cite{EILENBERG1967452}.

\begin{defi}[Recognizable $T$-languages]
Let $T$ be a monad on $\set$ and $\Sigma$ a finite set of symbols, called the {\it alphabet} in this context. A {\em $T$-language} over $\Sigma$ is 
a subset of $T\Sigma$. A $T$-algebra $\mathbf{A}=(A,\alpha)$ {\em recognizes} a $T$-language $L$ over $\Sigma$ iff there is a homomorphism $h:T\Sigma\to \mathbf{A}$, and a subset $S\subseteq A$, such that $L$ is the inverse image of $S$ along $h$.
A $T$-language is {\em recognizable} if it is recognized by some finite $T$-algebra.

For a fixed $T$, let $\rec_{T}$ denote the class of all recognizable $T$-languages.
\end{defi}

\begin{ex}\label{exmon}
For $T$ the free monoid monad (Monad~\ref{mon:list}), $\rec_{T}$ is the class of regular languages (or, equivalently, $\mso$-definable languages).\qed
\end{ex}

Recognizable languages are closed under Boolean operations.

\begin{prop}\label{prop1}
For any monad $T$, the class $\rec_{T}$ is closed under binary unions, binary intersections and complement, for a fixed alphabet.
\end{prop}

\begin{proof}
Let $L_i$ be a language that is recognized by a finite $T$-algebra $\mathbf{A}_i$ with $S_i\subseteq A_i$ through homomorphisms $h_i:T\Sigma\to \mathbf{A}_i$, $i=1,2$. Then:
\begin{itemize}
\item $L_1\cap L_2$ is recognized by $\mathbf{A}_1\times \mathbf{A}_2$ through the homomorphism $\langle h_1,h_2\rangle$ and the subset
\[
	\{(a_1,a_2) \mid a_1\in S_1 \mbox{ and } a_2\in S_2\};
\]
\item $L_1\cup L_2$ is recognized by $\mathbf{A}_1\times \mathbf{A}_2$ through the homomorphism $\langle h_1,h_2\rangle$ and the subset
\[
	\{(a_1,a_2) \mid a_1\in S_1 \mbox{ or } a_2\in S_2\};
\]
\item $T\Sigma\setminus L_1$ is recognized by $\mathbf{A}_1$ through $h_1$ and $A_1\setminus S_1$.
\end{itemize}
\end{proof}

Let $h:T\Sigma\to T\Gamma$ be a $T$-algebra homomorphism. If a language $L\subseteq T\Gamma$ is recognizable then the inverse image $\overleftarrow{h}(L)\subseteq T\Sigma$ is recognizable as well: it is recognized by the same $T$-algebra and subset as $L$. So recognizable languages are closed under inverse images of $T$-algebra homomorphisms. One may ask the same question about direct images: if $L\subseteq T\Sigma$ is recognizable,
is $\overrightarrow{h}(L)\subseteq T\Gamma$ recognizable as well? This question will be the our main technical focus in the following.

We call an algebra homomorphism $h:T\Sigma\to T\Gamma$ a {\it letter-to-letter homomorphism} if $h=Tf$ for some function $f:\Sigma\to \Gamma$. We call it {\it surjective} if $f$ is surjective\footnote{All functors on $\set$ preserve surjective functions since every surjective function $f$ has a right inverse $g$ in the sense that $f\circ g=\id$, then we apply $T$ to the last equation to obtain that $Tf$ is also surjective. Conversely, under the assumption that each component of $\eta$ is injective, we have that $T$ is faithful and faithful functors reflect epimorphisms (=surjections in the category $\set$). To show that $T$ is faithful assume $Tf=Tg$ for $f,g:X\to Y$ then
\[
	\eta_Y\circ g=Tg\circ \eta_X=Tf\circ \eta_X=\eta_Y\circ f
\]
which implies $g=f$ since $\eta_Y$ is injective. There will be no confusion between $Tf$ being surjective and $f$ being surjective since all the monads we consider are such that $\eta_\Sigma$ is injective for each $\Sigma$.}. Proposition \ref{propclosuremso} motivates the following definition \cite[Sec. 6]{boj1}.

\begin{defi}
Let $T$ be a monad on $\set$. Let $\mathcal{L}$ be a family of $T$-languages over finite alphabets. We define $\mso(\mathcal{L})$ as the least class of $T$-languages over finite alphabets that contains $\mathcal{L}$ and is closed under Boolean operations, inverse images of homomorphisms, and direct images of surjective letter-to-letter homomorphisms.
\end{defi}

\begin{ex}
Proposition \ref{propclosuremso} means that:
\[
	\mso(\{0^*,0^*10^*\})=\rec_{(-)^*}
\]
where both $0^*$ and $0^*1^*$ are considered as $(-)^*$-languages over $\{0,1\}$.
\qed
\end{ex}

Now we are ready to state the main technical question of this paper, which is a necessary step towards  a result similar to Proposition~\ref{propclosuremso} for other monads, namely:
\begin{quote}
	$\bullet$ When does $\mathcal{L}\subseteq \rec_{T}$ imply $\mso(\mathcal{L})\subseteq \rec_{T}$?
\end{quote}

By Proposition~\ref{prop1}, $\rec_{T}$ is closed under Boolean operations. Trivially, inverse images of recognizable languages are also recognizable. Thus the question is reduced to:
\begin{quote}
	$\bullet$ When is $\rec_{T}$ closed under direct images of surjective letter-to-letter homomorphisms?
\end{quote}
Asking this specific question is, to some extent, a design decision. In the classical setting of finite words, recognizable languages are closed under direct images along arbitrary homomorphisms, so one may reasonably ask for a stronger closure property: under arbitrary homomorphisms, or perhaps under arbitrary (i.e. possible non-surjective) letter-to-letter homomorphisms. We will look at these variants of the question in Section~\ref{sec:related}.

In any case, one may wonder whether perhaps the question (say, in the weak version as stated above) has an affirmative answer for {\em every} monad. Indeed, this was mistakenly claimed as a fact in~\cite[Lemma 6.2]{boj1}. However, as the following counterexample shows, the situation is not so simple.

\begin{monad}\label{mon:x3-x2}
Let $T$ be the free monoid monad quotiented by the additional equation
\begin{equation}\label{eq:burnside}
	x\mydot x\mydot x = x\mydot x.
\end{equation}
The congruence (which we denote $\sim$) induced on $\Sigma^*$ by~\eqref{eq:burnside} has been the subject of some research: Brzozowski in~\cite{brzozowski80} asked whether all equivalence classes of this congruence are regular languages, and the question
has remained open ever since~(see~\cite{pin17}). We do not need to answer that question for our purposes.

Elements of $T\Sigma$ are $\sim$-equivalence classes of finite $\Sigma$-words, and so a $T$-language over $\Sigma$ can be identified with a $\sim$-closed language of $\Sigma$-words. It is easy to see that a $T$-language is recognizable if and only if its associated language of $\Sigma$-words is regular.

Let $\Gamma=\{a,b,c\}$ and $\Sigma=\Gamma\cup\{0,1\}$. Consider the language of finite $\Sigma$-words:
\[
	L = \Gamma^*0\Gamma^*1.
\]
It is obviously regular. It is also $\sim$-closed. Indeed, assume that
\[
	u0v1=w
\]
for some $u,v\in\Gamma^*$ and $w\in\Sigma^*$. This means that there is a finite sequence of applications of~\eqref{eq:burnside} that transforms $u0v1$ into $w$. If this sequence is nonempty then its first step must detect a square, that is, a word of the form $xx$, in $u0v1$. (The other option is that a cube $xxx$ must be present, which implies a square anyway). Since $0$ and $1$ occur in that word only once each, it is clear that the square subword must occur entirely within $u$ or within $v$. This means that the result of applying~\eqref{eq:burnside} to $u0v1$ still belongs to $L$. By induction, also $w$ must belong to $L$.

We have thus proved that $L$, considered as a subset of $T\Sigma$, is $T$-recognizable.

Now consider $\Sigma'=\Gamma\cup\{0\}$ and let $h:\Sigma\to\Sigma'$ map $1$ to $0$ and act identically on all other letters. The direct image of $L$ along $h$, construed (as it should be) as a subset of $T\Sigma'$, corresponds to
a language $K$ of $\Sigma'$-words which arises as a $\sim$-closure of the language $\Gamma^*0\Gamma^*0$. (Note that the language itself is not $\sim$-closed, as e.g. $a0a0\sim a0a0a0$.)

We shall show that $K$ is not regular, which will imply that the direct image of $L$ under $h$ is not recognizable.

First, it is a well-known fact that since $\Gamma$ has $3$ elements, the language $\Gamma^*$ contains infinitely many square-free words. Moreover, two distinct square-free words cannot be $\sim$-equivalent, since the equation~\eqref{eq:burnside} cannot be applied to a square-free word.

Were the language $K$ regular, its Myhill-Nerode congruence would have a finite index, and so there would be two distinct square-free words $v,w\in\Gamma^*$ which are Myhill-Nerode equivalent with respect to $K$.

Thanks to~\eqref{eq:burnside} we have $v0v0v0\in K$
so, since $v$ and $w$ are Myhill-Nerode equivalent, also
\[
	v0w0v0\in K.
\]
But it is easy to check that if $v$ and $w$ are square-free and $v\neq w$ then $v0w0v0$ is also square-free, and so it is not $\sim$-equivalent to any word other than itself. We arrive at a contradiction, and so $K$ cannot be regular.
\end{monad}

More counterexamples of this kind are presented in Section~\ref{sec:counterexamples}. But first, now knowing that our main technical question is non-trivial, let us study some conditions on the monad $T$ that guarantee an affirmative answer to it.

\section{Simple sufficient conditions}\label{sec:suffcond}

In this section, we present two conditions on a monad $T$ that guarantee that the class $\rec_{T}$ is closed under direct images of surjective letter-to-letter homomorphisms. Another, more elaborate sufficient condition will be presented in Section~\ref{sec:suffcond-jacobs}.

\subsection{Monads that preserve finiteness}

One straightforward case where direct images of recognizable languages are also recognizable, is the case where the functor part of the monad preserves finiteness, i.e., maps finite sets to finite sets. Trivially, then, any language over $\Sigma$ is recognized by the identity homomorphisms into the finite algebra $T\Sigma$, so in particular any direct image is recognizable.

Examples of monads with this property include: the powerset monad and its variants, the double contravariant powerset monad $2^{2^{\Sigma}}$, monads for idempotent semigroups, idempotent monoids,  distributive lattices, Boolean algebras, semimodules over a finite semiring, and, in general, any {\it locally finite variety}. A variety is {\it locally finite} if every finitely generated algebra is finite.

For a given equational theory, checking if a finitely generated algebra is finite could be a challenging problem. For instance, the fact that a finitely generated idempotent semigroup is finite is a nontrivial fact, with multiple proofs published in different papers, (see e.g.~\cite{greenrees,brown}). Another interesting case are distributive lattices, where the exact size of the free algebra on $n$ generators is still unknown for values of $n>8$. 

\subsection{Monads with a Mal'cev term}\label{sec:malcev}

Equational theories over finite signatures that have a Mal'cev term induce monads for which direct images of recognizable languages along surjective letter-to-letter homomorphisms are also recognizable. A Mal'cev term is defined as follows \cite[II.12]{bys}.

\begin{defi}
A ternary term $p(x,y,z)$ on an equational class $V$ is called a {\it Mal'cev term} if the following identities hold in $V$:
\begin{equation}\label{eq:malcev}
	p(x,x,y)= y = p(y,x,x).
\end{equation}
\end{defi}
Equational classes with a Mal'cev term are {\it congruence-permutable} \cite[Theorem II.12.2]{bys}, i.e., any two congruences commute in them. Examples of such classes include:

\begin{monad}\label{mon:grp}
The {\em free group monad}, presented by a binary operation $\mydot$, a unary operation $(-)^{-1}$ and a constant operation $1$, subject to the usual group axioms. A Mal'cev term $p$ is given by $p(x,y,z)=x\mydot y^{-1}\mydot z$.
\end{monad}

Similarly, Mal'cev terms can be found for any equational class that has a group reduct, including abelian groups, rings, vector spaces over a field and algebras over a field. In this case, if the group reduct is given in additive notation, then the term $p$ is given by $p(x,y,z)=x-y+z$.

\begin{monad}\label{mon:qgrp}
The {\em free quasigroup monad}, presented by three binary operations $/$, $\mydot$ and $\backslash$ subject to equations:
\[
	x\backslash(x\mydot y)=x \qquad (x\mydot y)/y = x \qquad x\mydot(x\backslash y)=y \qquad (x/y)\mydot y=x.
\]
A Mal'cev term is $p(x,y,z) = (x/(y\backslash y))\cdot (y\backslash z)$.
\end{monad}

\begin{monad}\label{mon:ba}
The {\em free Boolean algebra monad}, presented by binary operations $\lor$ and $\land$, a unary operation $\neg$ and constants $0$ and $1$, subject to familiar equations. A Mal'cev term is
\[p(x,y,z) =(x\land z)\lor (x\land \lnot y \land \lnot z)\lor (\lnot x\land \lnot y \land z).\]
\end{monad}

\begin{monad}\label{mon:ha}
The {\em free Heyting algebra monad}, presented by binary operations $\lor$, $\land$ and $\to$, and constants $0$ and $1$, subject to the standard equations of Heyting algebras. A Mal'cev term is
\[p(x,y,z)=((x\to y)\to z)\land ((z\to y)\to z)\land (x\lor z).\]
\end{monad}
Note that, apart from the case of Boolean algebras, in all the above examples finitely generated free algebras are infinite.

\begin{prop}\label{prop:malcev}
For a monad $T$ presented by a set of equations that admits a Mal'cev term, the class $\rec_{T}$ is closed under direct images of surjective letter-to-letter homomorphisms.
\end{prop}

Before proceeding with the proof, we recall the standard notion of a congruence. Let $\mathcal{F}=\{f_i\}_{i\in I}$ be a type of algebras, where each operation symbol $f_i$ has arity $n_i$, and let $A$ be an $\mathcal{F}$--algebra. A binary relation $\theta$ on $A$ is a {\em congruence on $A$} if it is an equivalence relation and has the compatibility property in the following sense:
\begin{itemize}
    \item For every $i\in I$, $(a_1,b_1),\ldots,(a_{n_i},b_{n_i})\in \theta$ implies $(f_i(a_1,\ldots,a_{n_i}),f_i(b_1,\ldots,b_{n_i}))
    \in \theta$.
\end{itemize}
Given such a congruence, the quotient $A/\theta$ has an $\mathcal{F}$-algebra structure (see e.g.~\cite[II.5]{bys} for more properties of congruences and quotient algebras).

\begin{proof}
Let $L\subseteq T\Sigma$ be recognized by a homomorphism $h:T\Sigma\to A$ through a subset $S\subseteq A$ as $L=\overleftarrow{h}(S)$, where $A$ is a finite $T$-algebra. Without loss of generality, we may assume that $h$ is a surjective homomorphism. Let $g:\Sigma\to\Gamma$ be a surjective function. Define the relation $\approx$ on $T\Gamma$ as
\[
{\approx}=\{(Tg(u),Tg(v))\mid u,v\in T\Sigma,\ h(u)=h(v)\}.
\]
Then $\approx$ is a congruence on $T\Gamma$, which implies that $T\Gamma/{\approx}$ is a $T$-algebra. Indeed, the relation $\approx$ is clearly reflexive, symmetric and has the compatibility property (the latter because $h$ is a homomorphism). To prove transitivity, let $p(x,y,z)$ be a Mal'cev term and assume $a\approx b$ and $b\approx c$, then
\[a=p(a,b,b)\approx p(b,b,c)=c\] since $\approx$
 is reflexive and has the compatibility property.

 Let $\tau:T\Gamma \to T\Gamma/{\approx}$ be the canonical quotient map. Define a function $e:A\to T\Gamma/{\approx}$ by:
\[
	e(h(u)) = \tau(Tg(u)) \qquad \text{for } u\in T\Sigma;
\] 
this is well defined since $h$ is surjective and by the definition of $\approx$ and $\tau$. Moreover, $e$ is surjective because both $Tg$ and $\tau$ are; as a result, the $T$-algebra $T\Gamma/{\approx}$ is finite. 

We show that the direct image of $L$ along $Tg$ is recognized by the homomorphism $\tau$ to $T\Gamma/{\approx}$ through the subset $\overrightarrow{e}(S)\subseteq T\Gamma/{\approx}$. In other words, we will show that
\[
	\overrightarrow{Tg}\left(\overleftarrow{h}(S)\right) = \overleftarrow{\tau}\left(\overrightarrow{e}(S)\right).
\]
To this end, for any $w\in T\Gamma$ calculate:
\begin{align*}
w\in \overrightarrow{Tg}\left(\overleftarrow{h}(S)\right) &\iff
\exists v\in T\Sigma.\ Tg(v)=w,\ h(v)\in S \\
& \iff \exists u\in T\Sigma.\ Tg(u)\approx w,\ h(u)\in S \\
& \iff \exists u\in T\Sigma.\ \tau(Tg(u))=\tau(w),\ h(u)\in S \\
& \iff \exists u\in T\Sigma.\ e(h(u))=\tau(w),\ h(u)\in S \\
& \iff \exists a\in S.\ e(a)=\tau(w) \iff w\in\overleftarrow{\tau}\left(\overrightarrow{e}(S)\right)
\end{align*}
which finishes the proof.
\end{proof}

\begin{rem}
The key step in the above proof is the general observation that in a congruence-permutable variety, every reflexive and symmetric relation with the compatibility property is a congruence~\cite[Proposition 3.8.]{Hutchinson}. Reflexive, symmetric relations with the compatibility property are called {\it tolerances}.
\end{rem}

In Section~\ref{sec:not-quite-malcev} we shall see that relaxing the Mal'cev condition even slightly does not guarantee recognizable languages to be preserved under direct images along surjective letter-to-letter homomorphisms.

\section{Weakly epi-cartesian monads}\label{sec:suffcond-jacobs}
In this section, we present a class of monads -- called weakly epi-cartesian monads --  where a powerset construction can be used to prove that recognizable languages are closed under direct images of letter-to-letter homomorphisms. This class includes the classical case of finite words (Monad~\ref{mon:list}).

We begin with an intuitive description of the powerset construction. Consider a monad $T$ and a language $L \subseteq T \Sigma$ which is recognized by a homomorphism
$h : T \Sigma \to \mathbf A$
for some $T$-algebra $\mathbf A=(A,\alpha)$. For  a surjective function on the alphabet $f : \Sigma \to \Gamma$, we want to show that the direct image of $L$  under the letter-to-letter homomorphism $
T f : T \Sigma \to T \Gamma$
is also recognizable.

A natural idea is to consider a powerset algebra, where the universe is the family of nonempty subsets of $A$, and the product operation  is defined using:
\begin{align*}
    t \in T \PP\!A \quad \mapsto \quad \st{ \alpha(s) : s \in TA \underbrace{\text{ can be obtained from $t$ by choosing an element from each set}}_{\text{ can be made precise, see~\eqref{eq:jl} in Definition~\ref{def:jl}} }}.
\end{align*}
There are two issues that need to be addressed here. First, one must check that the powerset algebra is indeed a $T$-algebra, i.e., that it satisfies the axioms of Definition~\ref{def:eilenberg-moore}. The second issue is finding a homomorphism into the powerset algebra that recognizes the direct image of $L$ under $Tf$. The natural idea is to consider the  function
\begin{align*}
t \in T\Gamma  \qquad \mapsto \qquad \st{h(s) : s \in T \Sigma\ \text{ satisfies }T f(s)=t };
\end{align*}
however, it is not immediately clear that this function is a homomorphism. In fact this does not happen in every monad, contrary to what was claimed in~\cite[Lemma 6.2]{boj1}. For example, in the monad of groups (Monad~\ref{mon:grp}) the powerset algebra does not satisfy, in general, the axiom $x\cdot x^{-1}=e$.

The goal of this section is to establish a condition on the monad $T$ which ensures that the above powerset construction works as expected.

\subsection{Definitions}\label{sec:wecmonad}

First, a few standard categorical definitions. A commuting diagram of sets and functions:
\begin{equation}\label{eq:wpb}
\vcenter{\xymatrix{
P\ar[r]^h\ar[d]_k & X\ar[d]^f \\
Y\ar[r]_g & Z
}}
\end{equation}
is called a {\em weak pullback} if, for each $x\in X$ and $y\in Y$ such that $f(x)=g(y)$, there is some $p\in P$ such that $h(p)=x$ and $k(p)=y$. If $p$ is unique for each such $x$ and $y$ then the square is a {\em pullback}. For every two functions $f$ and $g$ as above a canonical pullback exists and is defined by
\begin{equation}\label{eq:pb}
	P = \{(x,y) \mid f(x)=g(y)\},
\end{equation}
with $h=\pi_1$ and $k=\pi_2$, projections from $P$ into $X$ and $Y$.

A functor $T$ {\em preserves weak pullbacks} iff applying $T$ to everything in a weak pullback as in~\eqref{eq:wpb} results in a weak pullback again. A routine categorical argument shows that one may equivalently require only that (canonical) pullbacks are mapped to weak pullbacks; preservation of all other weak pullbacks follows. So, $T$ preserving weak pullbacks means that for all functions $f:X\to Z$ and $g:Y\to Z$, and for any $t\in TX$ and $r\in TY$ such that $Tf(t)=Tg(r)$, there exists some $p\in TP$ (for $P$ as in~\eqref{eq:pb}) such that $T\pi_1=t$ and $T\pi_2=r$.

Intuitively, this means that whenever two terms $t$ and $r$ are sufficiently similar to be equated by equating some variables, then there is some term $p$, built of compatible pairs of variables, that projects to $t$ and $r$. Many monads, such as the free monoid monad and the powerset monad and its variants, have this property. The following examples show how the property may fail.

\begin{ex}
The free group monad (Monad~\ref{mon:grp}) does not preserve weak pullbacks: the square
\[
\xymatrix{
\emptyset\ar[r]^h\ar[d]_k & \emptyset\ar[d]^f \\
\{a,b\}\ar[r]_g & \{c\}
}
\]
(with $f,g,h,k$ the unique functions of their type) is a pullback, and terms $t=1\in T\emptyset$ and $r=a\mydot b^{-1}\in T\{a,b\}$ are equated by $Tf$ and $Tg$, but there is no term $p\in T\emptyset$ that would be mapped to $r$ by $Tk$.
\qed
\end{ex}

\begin{ex}
Monad~\ref{mon:x3-x2}, i.e., the free monoid monad subject to the additional axiom $x\mydot x\mydot x = x\mydot x$, does not preserve weak pullbacks. Indeed, the square
\[
\xymatrix{
\{a,b\}\times\{c,d\}\ar[r]^-{\pi_1}\ar[d]_{\pi_2} &\{a,b\}\ar[d]^f \\
\{c,d\}\ar[r]_-g & \{e\}
}
\]
(with $f,g$ the unique functions of their type) is a pullback, and the terms $t=a\mydot b\in T\{a,b\}$ and $r=c\mydot d\mydot c\in T\{c,d\}$ are equated by $Tf$ and $Tg$, but there is no term $p\in T(\{a,b\}\times\{c,d\})$ that would project to $t$ and $r$.
\qed
\end{ex}

Given $\set$-functors $F$ and $G$, a natural transformation $\alpha:F\Rightarrow G$ is {\em weakly cartesian} if for every function $f:X\to Y$ the naturality square
\[\xymatrix{
FX\ar[r]^{\alpha_X}\ar[d]_{Ff} & GX\ar[d]^{Gf} \\
FY\ar[r]_{\alpha_Y} & GY
}\]
is a weak pullback. A weaker property is being {\em weakly epi-cartesian}, where only naturality squares for surjective $f$ are required to be weak pullbacks.

We can now formulate our main property of interest:

\begin{defi}\label{def:wcmon}
A monad $(T,\eta,\mu)$ is {\em weakly cartesian} if $T$ preserves weak pullbacks and $\eta$ and $\mu$ are weakly cartesian. It is {\em weakly epi-cartesian} if $T$ preserves weak pullbacks and $\eta$ and $\mu$ are weakly epi-cartesian.
\end{defi}

Weakly cartesian monads have been studied in the literature~\cite{weber04,chj14,fp18}. The notion of a weakly epi-cartesian monad seems to be new.

The requirement that the naturality square
\[\xymatrix{
	X\ar[r]^{\eta_X}\ar[d]_f & TX\ar[d]^{Tf} \\
	Y\ar[r]_{\eta_Y} & TY
}\]
is a weak pullback amounts to saying that whenever a unit term $\eta_Y(y)$ is obtained as a value of $Tf$ applied to some $t\in TX$, then $TX$ must itself be a unit term $\eta_X(x)$ for some $x$ such that $f(x)=y$. So weak (epi-)cartesianness of $\eta$ intuitively means that one cannot obtain a unit term by equating variables in a non-unit term. This holds for many monads, including the free monoid monad. Notably, the property fails for the powerset monad, where the non-unit term $\{a,b\}\in{\cal P}\{a,b\}$ is mapped to the unit $\{c\}\in{\cal P}\{c\}$ by ${\cal P}f$, for the unique function $f:\{a,b\}\to \{c\}$.

Let us now look at the multiplication $\mu$ being weakly (epi)-cartesian. To say that the naturality square
\[\xymatrix{
	TTX\ar[r]^{\mu_X}\ar[d]_{TTf} & TX\ar[d]^{Tf} \\
	TTY\ar[r]_{\mu_Y} & TY
}\]
is a weak pullback means that, for any term $t\in TX$, if a term $Tf(t)$ can be decomposed as ``term of terms'' $\theta\in TTY$ (so that $Tf(t)$ is the flattening of $\theta$), then the term $t$ itself has a similar decomposition. So, intuitively, weak (epi)-cartesianness of $\mu$ means that equating variables in a term does not introduce essentially new ways of decomposing the term.

\begin{thm}\label{prop:ewcm}
If $T$ is a weakly epi-cartesian monad, then $\rec_{T}$ is closed under direct images along surjective letter-to-letter homomorphisms.
\end{thm}
\begin{proof}
See Section~\ref{sec:theproof}.
\end{proof}

\subsection{Examples}

The three conditions of weak (epi-)cartesianness may fail in various configurations. For example, the powerset monad and its variants (Monads~\ref{mon:power},~\ref{mon:npower} and~\ref{mon:finpower}) preserve weak pullbacks and have weakly cartesian multiplication, but their unit is not weakly epi-cartesian. On the other hand, Monad~\ref{mon:x3-x2}
does not preserve weak pullbacks, its unit is weakly cartesian and its multiplication is not weakly epi-cartesian.
Further examples include:

\begin{monad}\label{mon:almost-malcev}
Consider the monad $(T,\eta,\mu)$ associated to the equational theory of a single ternary function symbol $p$ and a single unary symbol $s$, with axioms $p(x,x,y)=s(y)=p(y,x,x)$. Then $T$ does not preserve weak pullbacks, $\eta$ is weakly cartesian, and $\mu$ is weakly cartesian.

We shall study this monad in more detail in Section~\ref{sec:not-quite-malcev}.
\end{monad}

\begin{monad}\label{mon:ffe}
Consider the monad $(T,\eta,\mu)$ associated to the equational theory on the signature $\{f,e\}$, where $f$ is a unary function symbol and $e$ is a constant, with the only axiom $f(f(x))=e$. Then $T$ preserves weak pullbacks, $\eta$ is weakly cartesian, $\mu$ is weakly epi-cartesian, but $\mu$ is not weakly cartesian, cf. \cite[Example 3.4]{chj14}.

This example distinguishes between the notions of weakly cartesian and weakly epi-cartesian monads.
\end{monad}

However, many monads are weakly epi-cartesian, as the following examples show.

\medskip

A rich source of weakly cartesian monads (and therefore of weakly epi-cartesian monads) are equational theories whose axioms are given by regular linear equations. A term is called {\it linear} if no variable appears in it more than once. An equation is called {\it regular} if each variable appears on the left-hand side if and only if it appears on the right-hand side. An equation is {\it regular linear} if it is regular and both sides of it are linear. For instance, the equations $x\mydot (y\mydot z)=(x\mydot y)\mydot z$, $x\mydot e=x$, and $e\mydot x=x$, which axiomatize the variety of monoids, are regular linear. The commutativity axiom $x\mydot y=y\mydot x$ is also regular linear. On the other hand, the group axiom $x\mydot x^{-1}=e$ and the idempotence axiom $x\mydot x=x$ are not regular linear.

It is not difficult to check that a monad presented by a set of regular linear equations is weakly cartesian.
In fact, every such monad is {\em analytic}: in addition to being weakly cartesian, its functor weakly preserves {\em wide} pullbacks. This is a part of a more general framework, see~\cite{SzZ14}.

In particular, the monads of semigroups, commutative semigroups, monoids, commutative monoids and the variety of all algebras for any given signature, are all weakly cartesian, hence also weakly epi-cartesian.

\begin{ex}
$\rec_{T}$ is closed under direct images of surjective letter-to-letter homomorphisms for the monads of semigroups, commutative semigroups, monoids, commutative monoids and the variety of all algebras for any given signature.\qed
\end{ex}

The next four example monads are not finitary, so they do not have finitary equational presentations. Nevertheless, they are all weakly cartesian.

\begin{monad}[A monad for $\omega$-words]\label{mon:omega-words} Consider the following extension of the list monad to infinite words. The monad maps an alphabet $\Sigma$ to the set $\Sigma^+ \cup \Sigma^\omega$ of nonempty words of length at most $\omega$.  The unit is the same as in the list monad, while the multiplication operation is defined by
	\[
	\mu_X(w_1w_2\cdots) =\begin{cases}
		w_1w_2\cdots &\text{if all words $w_1,w_2,\ldots$ are finite}\\
		w_1w_2\cdots w_k &\text{if $w_k$ is the first infinite word among $w_1,w_2,\ldots$}\\
	\end{cases}
	\]
	Algebras for this monad are essentially the same as $\omega$-semigroups~\cite[Section II.4]{perrinInfiniteWordsAutomata2004}, which are known to recognize the same languages as B\"uchi automata on $\omega$-words.
	This monad is weakly  epi-cartesian, so by Theorem~\ref{prop:ewcm} the recognizable languages are closed under  direct images along surjective letter-to-letter homomorphisms, thus showing that all \msotext definable languages are recognized by $\omega$-semigroups.
	\end{monad}

\begin{monad}[Countable linear orders] Monad~\ref{mon:omega-words} can be generalised from $\omega$-words to other labelled linear orders, where the order type of the positions is not necessarily $\omega$. Consider the monad where $T\Sigma$ is the set of  countable linear orders labelled by $\Sigma$, up to isomorphism. The unit and multiplication are defined in the natural way (contrary to Monad~\ref{mon:omega-words}, there is no need to truncate).
	For example, $T \st{a,b}$ contains the following element: the linear order of the rational numbers, labelled  so that the positions with label $a$ are dense and the same is true for the positions with label $b$. Using a back-and-forth argument, one can show that element described above is unique up to isomorphism. Shelah showed that satisfiability is decidable for \msotext over this monad~\cite[Theorem 6.2]{shelah_monadic_1975}, while  Carton, Colcombet and Puppis~\cite[Proposition 3 and Theorem 3]{carton_regular_2011} showed that the recognizable languages are exactly the ones that are definable in \msotext. This monad is weakly epi-cartesian, and therefore Theorem~\ref{prop:ewcm} implies that recognizable languages are closed under direct images of surjective letter-to-letter homomorphisms, which in turn implies that \msotext definable languages are recognizable~\cite[Proposition 3]{carton_regular_2011}.
\end{monad}

In fact, as far as closure under direct images is concerned, there is nothing special about countable linear orders, as shown by the following monad.
\begin{monad}
	\label{mon:uncountable} One can also consider a variant of the previous monad, but for labelled linear orderings of cardinality at most continuum. This monad is also weakly epi-cartesian, and therefore  \msotext definable languages are recognizable. There is, however, a price to pay for considering uncountable orders: the satisfiability problem for \msotext is undecidable for this monad, as proved by Shelah~\cite[Theorem 7]{shelah_monadic_1975}.
\end{monad}

There are also interesting monads that lie between $\omega$-word and all countable linear orders.

\begin{monad}
	Consider the countable linear orders which are scattered, i.e.~do not contain any rational sub-ordering. This monad was studied implicitly by Carton and Rispal, where the Eilenberg-Moore algebras for the monad were called $\diamond$-algebras in~\cite[Definition 6]{rispalComplementationRationalSets2005}. This monad is weakly epi-cartesian, and therefore Theorem~\ref{prop:ewcm} implies that \msotext definable languages are necessarily recognizable (this was already known in~\cite{rispalComplementationRationalSets2005}).
\end{monad}

Some monads do not admit any presentation by regular linear equations, yet they still are weakly (epi-)cartesian. For example:

\begin{monad}[A left-idempotent operation]\label{mon:xxy-xy}
For a more unusual example of a weakly cartesian monad, consider a monad $T$ presented by an equational theory with a single binary symbol $\mydot$ and with a single (non-linear) equation:
\begin{equation}\label{eq:xxy-xy}
\vcenter{\xymatrix@=.6em{
& \mydot\ar@{-}[ld]\ar@{-}[rd] \\
x & & \mydot\ar@{-}[ld]\ar@{-}[rd] \\
& x & & y
}}
\quad=\quad
\vcenter{\xymatrix@=.6em{
& \mydot\ar@{-}[ld]\ar@{-}[rd] \\
x & & y
}}.
\end{equation}
To prove that $T$ preserves weak pullbacks, consider any pullback diagram
\[
\xymatrix{
P\ar[r]^{\pi_1}\ar[d]_{\pi_2} & X\ar[d]^f \\
Y\ar[r]_g & Z,
}
\]
where $P = \{(x,y)\mid f(x)=g(y)\}$ and $\pi_1$ and $\pi_2$ are the projections from $P$. Then, for any $r\in TX$, $s\in TY$ and $t\in TZ$ such that
\[
	Tf(r)=t=Tg(s),
\]
we need to find $p\in TP$ such that $T\pi_1(p)=r$ and $T\pi_2(p)=s$.

Formally, elements of $TZ$ are equivalence classes of binary trees with elements of $Z$ in leaves; each equivalence class can be represented by its unique smallest element: one where the pattern on the left-hand side of~\eqref{eq:xxy-xy} does not appear. We therefore only consider $r$, $s$ and $t$ of this form, and we proceed by induction on the size of $t$.

For the base case, if $t$ is a single letter (i.e. $t\in Z$) then obviously $r\in X$ and $s\in Y$, therefore we can put $p=(r,s)$.

For the inductive step, let
\[t\quad=\quad\vcenter{\xymatrix@=.6em{
& \mydot\ar@{-}[ld]\ar@{-}[rd] \\
t' & & z}}\]
for some $t',z\in TZ$. Since $Tf(r)=t$, it follows that
\[
r \quad=\quad \vcenter{\xymatrix@=.6em{
& \mydot\ar@{-}[ld]\ar@{-}[rd] \\
r_1 & & \mydot\ar@{-}[ld]\ar@{.}[rd] \\
& r_2 & & \mydot\ar@{-}[ld]\ar@{-}[rd] \\
& & r_n & & x
}}
\]
for some $r_1,\ldots,r_n,x\in TX$ such that $Tf(r_i)=t'$ and $Tf(x)=z$. Similarly,
\[
s \quad=\quad \vcenter{\xymatrix@=.6em{
& \mydot\ar@{-}[ld]\ar@{-}[rd] \\
s_1 & & \mydot\ar@{-}[ld]\ar@{.}[rd] \\
& s_2 & & \mydot\ar@{-}[ld]\ar@{-}[rd] \\
& & s_m & & y
}}
\]
for some $s_1,\ldots,s_m,y \in TY$ such that $Tg(s_i)=t'$ and $Tg(y)=z$.
Assume, without loss of generality, that $n\leq m$.

By the inductive assumption, for each $1\leq i\leq n$ and $1\leq j\leq m$ there is a $p_{i,j}\in TP$ such that $T\pi_1(p_{i,j})=r_i$ and $T\pi_2(p_{i,j})=s_j$. Also by the inductive assumption, there is some $q\in TP$ such that $T\pi_1(q)=x$ and $T\pi_2(q)=y$.

Now consider $p\in TP$ defined by:
\[
p \quad=\quad \vcenter{\xymatrix@=.6em{
& \mydot\ar@{-}[ld]\ar@{-}[rd] \\
p_{1,1} & & \mydot\ar@{-}[ld]\ar@{.}[rd] \\
& p_{2,2} & & \mydot\ar@{-}[ld]\ar@{-}[rd] \\
& & p_{n,n} & & \mydot\ar@{-}[ld]\ar@{.}[rd] \\
& & & p_{n,n+1} & & \mydot\ar@{-}[ld]\ar@{-}[rd] \\
& & & & p_{n,m} & & q.
}}
\]
It is easy to check that $T\pi_1(p)=r$ and $T\pi_2(p)=s$ as required. This completes the proof that $T$ preserves weak pullbacks.

It is easy to see that the unit of $T$ is cartesian, i.e., that the square
\[
\xymatrix{X\ar[r]^{\eta_X}\ar[d]_f & TX\ar[d]^{Tf} \\ Y\ar[r]_{\eta_Y} & TY}
\]
is a pullback for every $f:X\to Y$. Indeed, if $Tf(t)=\eta_Y(y)$ for $t\in TX$ and $y\in Y$ then there must be $t=\eta_X(x)$ for some (necessarily, unique) $x\in X$ such that $f(x)=y$.

It remains to be proved that the multiplication of $T$ is weakly cartesian, i.e., that
\[
\xymatrix{TTX\ar[r]^{\mu_X}\ar[d]_{TTf} & TX\ar[d]^{Tf} \\ TTY\ar[r]_{\mu_Y} & TY}
\]
is a weak pullback for every $f:X\to Y$. The proof is similar to the proof of the fact that $T$ preserves weak pullbacks.
For any $r\in TTY$, $s\in TX$ and $t\in TY$ such that
\[
	\mu_X(r)=t=Tf(s),
\]
we need to find $p\in TTX$ such that $TT(p)=r$ and $\mu_X(p)=s$, and the construction proceeds by induction on the size of $t$.
\end{monad}

\begin{monad}[A guarded-idempotent operation]
Consider the monad $T$ defined by the signature $\{\mydot, o\}$, where $\mydot$ is a binary operation symbol and $o$ is a unary operation symbol, and equations $x\mydot (y\mydot z)=(x\mydot y)\mydot z$ and $o(x\mydot x)=o(x)$.

To prove that $T$ preserves weak pullbacks, consider any pullback diagram
\[
\xymatrix{
P\ar[r]^{\pi_1}\ar[d]_{\pi_2} & X\ar[d]^f \\
Y\ar[r]_g & Z,
}
\]
where $P = \{(x,y)\mid f(x)=g(y)\}$ and $\pi_1$ and $\pi_2$ are the projections from $P$. Let $r\in TX$, $s\in TY$ and $t\in TZ$ such that $Tf(r)=t=Tg(s)$,
we need to find $p\in TP$ such that $T\pi_1(p)=r$ and $T\pi_2(p)=s$.

Elements of $TZ$ are equivalence classes of terms on $Z$; each equivalence class can be represented by its unique smallest element, up to associativity: one where the pattern $o(x\cdot x)$ does not appear. We therefore assume $r$, $s$ and $t$ of this form, and we proceed by induction on the size of $t$.

For the base case, if $t$ is a single letter (i.e. $t\in Z$) then obviously $r\in X$ and $s\in Y$, therefore we can put $p=(r,s)$.

For the inductive case, let $t\in TZ$. We have the following cases:
\begin{enumerate}[i)]
\item $t$ does not contain the function symbol $o$. In this case, $r$ and $s$ do not contain $o$ either. Therefore,
$t$ is of the form $t=z_1\cdot \cdots \cdot z_n$ which implies that $r$ and $s$ are of the form $r=x_1\cdot \cdots \cdot x_n$ and $s=y_1\cdot \cdots \cdot y_n$. Hence, we can take $p=(x_1,y_1)\cdot \cdots\cdot (x_n,y_n)$.
\item $t$ contains the function symbol $o$. Let $t=t_1\cdot o(t_2)\cdot t_3$ where $t_i\in TX$. Then, $r$ and $s$ are of the form $r=r_1\cdot o(r_2)\cdot r_3$ and $s=s_1\cdot o(s_2)\cdot s_3$, $r_i\in TX$ and $s_i\in TY$, with $Tf(r_1)=Tg(s_1)$, $Tf(o(r_2))=Tg(o(s_2))$ and $Tf(r_3)=Tg(s_3)$. By the induction hypothesis, there exists $p_1, p_3\in TP$ such that:
\[
	T\pi_1(p_1)=r_1,\quad T\pi_2(p_1)=s_1,\quad T\pi_1(p_3)=r_3,\text{ \ and \ \ }T\pi_2(p_3)=s_3,
\]
Now, since $Tf(o(r_2))=Tg(o(s_2))$ we have the following cases:
\begin{enumerate}[a)]
	\item $Tf(r_2)=Tg(s_2)$. Then, by the induction hypothesis, there exists $p_2\in TP$ such that $T\pi_1(p_2)=r_2$ and $T\pi_2(p_2)=s_2$.
	\item $Tf(r_2)=Tg((s_2)^{2^n})$ for some $n\geq 0$. Then, by the induction hypothesis, there exists $p_2\in TP$ such that $T\pi_1(p_2)=r_2$ and
		$T\pi_2(p_2)=(s_2)^{2^n}$.
	\item $Tf((r_2)^{2^n})=Tg(s_2)$ for some $n\geq 0$. Then, by the induction hypothesis, there exists $p_2\in TP$ such that
		$T\pi_1(p_2)=(r_2)^{2^n}$ and $T\pi_2(p_2)=s_2$.
\end{enumerate}
In each case, put $p=p_1\cdot o(p_2)\cdot p_3$.
\end{enumerate}
It is easy to check that, in each case, $T\pi_1(p)=r$ and $T\pi_2(p)=s$ as required. This completes the proof that $T$ preserves weak pullbacks.

To see that the unit of $T$ is cartesian, consider a function $f:X\to Y$, $t\in TX$ and $y\in Y$ such that $Tf(t)=\eta_Y(y)$. Then there must be $t=\eta_X(x)$ for some (necessarily, unique) $x\in X$ such that $f(x)=y$.

Finally, the proof that the multiplication $\mu$ is weakly cartesian is done by using a similar argument as the one that $T$ preserves weak pullbacks.
\end{monad}

\begin{monad}[Weakly epi--cartesian monad that is not weakly cartesian]
Consider the monad $T$ presented by an equational theory with two unary operation symbols $f$ and $g$, and a constant symbol $e$ subject to the equations:
\[
	f(f(x))=g(g(y))=e
\]
Note that $TX$ is infinite for every non-empty $X$ ($\{fg(x),fgfg(x),fgfgfg(x),\ldots\}\subseteq TX$ for $x\in X$). Also, the multiplication $\mu$ is not weakly cartesian since for the inclusion $\iota:\{x\}\to \{x,y\}$ we have that $T\iota(e)=\mu_{\{x,y\}}\Big(f\big(f(y)\big)\Big)$ but $f\big(f(y)\big)$ is not in the image of $TT\iota$. Nevertheless, $T$ is a weakly epi-cartesian monad, so $\rec_{T}$ is closed under direct images along surjective letter-to-letter homomorphisms.

\end{monad}

\subsection{Proof of Theorem~\ref{prop:ewcm}}\label{sec:theproof}

We start by recalling the standard concept of a distributive law between monads~\cite{beck}. Distributive laws are a standard tool for composing monads; we will use them to lift the powerset functor to the category of algebras for a monad.

\begin{defi}\label{def:dl}
Let $(T,\eta^T,\mu^T)$ and $(S,\eta^S,\mu^S)$ be monads. A {\em distributive law} of $(T,\eta^T,\mu^T)$ over $(S,\eta^S,\mu^S)$ is a natural transformation $\lambda: TS\Rightarrow ST$ that satisfies the following axioms:
\[\xymatrix{
T\ar[r]^{T\eta^S}\ar@/_3ex/[rd]_{\eta^ST}\ar@{}[rd]|-{(a)} & TS\ar[d]_{\lambda} & S\ar[l]_{\eta^TS}\ar@/^3ex/[ld]^{S\eta^T}\ar@{}[ld]|-{(b)} \\
& ST
}
\qquad
\xymatrix{
TTS\ar[rr]^{\mu^TS}\ar[d]_{T\lambda}\ar@{}[rrd]|-{(d)} & & TS\ar[d]^{\lambda} & & TSS\ar[ll]_{T\mu^S}\ar[d]^{\lambda S}\ar@{}[lld]|-{(c)} \\
TST\ar[r]_{\lambda T} & STT\ar[r]_{S\mu^T} & ST & SST\ar[l]^{\mu^ST} & STS\ar[l]^{S\lambda}
}\]
Sometimes one is interested in laws that satisfy only some of these axioms. In particular, a distributive law of the functor $T$ over the monad $(S,\eta^S,\mu^S)$ only satisfies the axioms $(a)$ and $(c)$. We will be interested in distributive laws of monads over the powerset or the nonempty powerset, i.e., $S={\cal P}$ or $S=\PP$.
\end{defi}

\begin{defi}[The Jacobs Law]\label{def:jl}
Let $T$ be a functor. For every set $X$ define a function $\JL_X: T\P X\to \P TX$ by:
\begin{equation}\label{eq:jl}
	\JL_X(t)=\{s \in TX \mid \exists r\in T({\in_X}) \text{ s.t. } T\pi_1(r)=s,\ T\pi_2(r)=t\},
\end{equation}
where ${\in_X}$ is the membership relation on $X$:
\[
	{\in_X} = \{(x,S) \mid x\in S\subseteq X\} 
\]
and $\pi_1:{\in_X}\to X$ and $\pi_2:{\in_X}\to \PP X$ are the canonical projections.
\end{defi}

It is easy to see that if $t\in T\PP X$ then $\lambda_X(t)$ is not empty. Indeed, pick any function $k:\PP X\to X$ such that $k(S)\in S$ for all $S\in\PP X$. Then
\[
	\langle k,\id\rangle : \PP X \to {\in_X}
\]
and the element $r=T\langle k,\id\rangle(t)\in T({\in_X})$ witnesses the fact that $\JL_X(t)$ is not empty. As  a result,
$\JL_X$ restricts to a function $\JL_X: T\PP X\to \PP TX$.

The following result is due to Bart Jacobs~\cite[Section 4]{jacobsTSC}.

\begin{prop}\label{prop:jacobs}
If a functor $T$ preserves weak pullbacks, then $\JL$ as in Def.~\ref{def:jl} is a natural transformation and it is a distributive law of the functor $T$ over the powerset monad $\P$ (and over the non-empty powerset monad $\PP$).
\qed
\end{prop}

For any sets $X$ and $Y$, there is a canonical partial order on the set of functions from $X$ to $\P Y$: we say that $f\leq g:X\to \P Y$ iff $f(x)\subseteq g(x)$ for all $x\in X$.

The following property of $\lambda$ will be useful later:

\begin{prop}\label{prop:Jacmon}
If $T$ preserves weak pullbacks then $\JL$ as in Def.~\ref{def:jl} is {\em monotone}: whenever $f\leq g:X\to \P Y$ then
\[
	\lambda_X\circ Tf \leq \lambda_X\circ Tg : TX\to \P TY.
\]
\end{prop}
\begin{proof}
Assume $f\leq g$, and pick some $t\in TX$ and $s\in \lambda_Y(Tf(t))$. This means that there is some $r\in T{\in_Y}$ such that
\[
	T\pi_1(r)=s \qquad\text{and}\qquad T\pi_2(r)=Tf(t).
\]
Consider the pullback
\[\xymatrix{
Q\ar[r]^{\pi_2}\ar[d]_{\id\times f} & X\ar[d]^f \\
\in_Y\ar[r]_{\pi_2} & \P Y
}\]
where
\[
	Q = \{(y,x)\mid x\in X,\ y\in f(x)\}.
\]
Since $T$ preserves weak pullbacks, the square
\[\xymatrix{
TQ\ar[r]^{T\pi_2}\ar[d]_{T(\id\times f)} & TX\ar[d]^{Tf} \\
T{\in_Y}\ar[r]_{T\pi_2} & T\P Y
}\]
is a weak pullback. By the properties of $r$, and by the definition of a weak pullback, there is some $q\in TQ$ such that
\[
	T(\id\times f)(q) = r \quad (\text{hence } T\pi_1(q)=T\pi_1(r)) \qquad\text{and}\qquad T\pi_2(q)=t.
\]
Since $f\leq g$, the function $\id\times g:Q\to {\in_Y}$ is well defined. Put $u=T(\id\times g)(q)\in T{\in_Y}$ and calculate:
\[
	T\pi_1(u) = T\pi_1(q) = T\pi_1(r) = s
\]
and
\[
	T\pi_2(u) = Tg(T\pi_2(q)) = Tg(t),
\]
hence $u$ is a witness to the fact that $s\in\lambda_Y(Tg(t))$ as required.
\end{proof}

We now extend Jacobs's Proposition~\ref{prop:jacobs} from functors to monads (cf.~Definition~\ref{def:wcmon}).

\begin{thm}\label{thm:distlaw}
If a monad $(T,\eta,\mu)$ is weakly cartesian, then $\JL$ as in Def.~\ref{def:jl} is a distributive law of the monad over the powerset monad $\P$. If it is weakly epi-cartesian, then $\JL$ is a distributive law of the monad over the non-empty powerset monad $\PP$.
\end{thm}
\begin{proof}
The proofs of both statements are almost identical, so we shall present only the proof of the first statement (the one for $\P$), and mark with $(\star)$ the two places where replacing $\P$ with $\PP$ makes a difference.

Since weakly (epi)-cartesian monads preserve weak pullbacks by definition, the naturality of $\lambda$ and axioms $(a)$ and $(c)$ in Definition~\ref{def:dl} follow from Proposition~\ref{prop:jacobs}.
What remains to be proved is that $\lambda$ satisfies axioms $(b)$ and $(d)$ in Definition~\ref{def:dl}.

For the axiom $(b)$, we need to prove that
\begin{equation}\label{eq:axiomb}
	\lambda_X(\eta_{\P X}(Z)) = \P\eta_X(Z) \qquad \text{for } Z\subseteq X.
\end{equation}
For the left-hand side, unfold from~\eqref{eq:jl}:
\begin{equation}\label{eq:Blhs}
	s\in\lambda_X(\eta_{\P X}(Z)) \iff \exists r\in T({\in_X})\text{ s.t. } T\pi_1(r)=s,\ T\pi_2(r)=\eta_{\P X}(Z).
\end{equation}
For the right-hand side:
\begin{equation}\label{eq:Brhs}
	s\in\P\eta_X(Z) \iff \exists z\in Z\text{ s.t. } s=\eta_X(z).
\end{equation}
For the right-to-left inclusion of~\eqref{eq:axiomb}, given $z$ as in~\eqref{eq:Brhs}, it is enough to put $r=\eta_{\in_X}(z,Z)$ and use naturality of $\eta$.

For the left-to-right inclusion of~\eqref{eq:axiomb}, assume $r$ as in~\eqref{eq:Blhs}. Since $T$ is weakly cartesian, the square
\[\xymatrix{
\in_X\ar[r]^{\pi_2}\ar[d]_{\eta_{\in_X}} & \P X\ar[d]^{\eta_{\P X}} \\
T{\in_X}\ar[r]_{T\pi_2} & T\P X
}\]
is a weak pullback. $(\star)$ If $\P$ is replaced by $\PP$, it is enough to assume that $T$ is weakly epi-cartesian. Indeed, the function $\pi_2:{\in_X}\to \PP X$ is surjective, so weak-epi-cartesianness is enough to conclude that the square is a weak pullback.

By the properties of $r$ stated in~\eqref{eq:Blhs}, and by the definition of weak pullback, there is some $z\in Z$ such that $\eta_{\in_X}(z,Z)=r$. By naturality of $\eta$ on $\pi_1:{\in_X}\to X$, we obtain
\[
	\eta_X(z) = \eta_X(\pi_1(z,Z)) = T\pi_1(\eta_{\in_X}(z,Z)) = T\pi_1(r) = s
\]
as required. This completes the proof of~\eqref{eq:axiomb}.

For the axiom (d), we need to prove that
\begin{equation}\label{eq:axiomd}
	\lambda_X(\mu_{\P X}(t)) = \P\mu_X(\lambda_{TX}(T\lambda_X(t))) \qquad \text{for } t\in TT\P X.
\end{equation}
For the left-hand side, unfold from~\eqref{eq:jl}:
\begin{equation}\label{eq:Dlhs}
	s\in\lambda_X(\mu_{\P X}(t)) \iff \exists r\in T({\in_X})\text{ s.t. } T\pi_1(r)=s,\ T\pi_2(r)=\mu_{\P X}(t).
\end{equation}
For the right-hand side:
\begin{align}\label{eq:Drhs}
s \in \P\mu_X(\lambda_{TX}(T\lambda_X(t))) &\iff \exists u\in\lambda_{TX}(T\lambda_X(t))\text{ s.t. }\mu_X(u)=s \nonumber \\
&\iff \exists v\in T({\in_{TX}})\text{ s.t. } \mu_X(T\pi_1(v))=s,\ T\pi_2(v) = T\lambda_X(t).
\end{align}

For the left-to-right inclusion of~\eqref{eq:axiomd}, assume $r$ as in~\eqref{eq:Dlhs}. Since $T$ is weakly cartesian, the square
\[\xymatrix{
TT{\in_X}\ar[r]^{TT\pi_2}\ar[d]_{\mu_{\in_X}} & TT\P X\ar[d]^{\mu_{\P X}} \\
T{\in_X}\ar[r]_{T\pi_2} & T\P X
}\]
is a weak pullback. $(\star)$ As before, if $\P$ is replaced by $\PP$, it is enough to assume that $T$ is weakly epi-cartesian since the function $\pi_2:{\in_X}\to \PP X$ is surjective.

By the properties of $r$ stated in~\eqref{eq:Dlhs}, and by the definition of weak pullback, there is some $w\in TT{\in_X}$ such that
\[
	\mu_{\in_X}(w)=r \qquad\text{and} \qquad TT\pi_2(w)=t.
\]

Consider a function $\kappa:T{\in_X}\to {\in_{TX}}$ defined by:
\[
	\kappa(u) = (T\pi_1(u),\lambda_X(T\pi_2(u))).
\]
This is well defined: indeed, $u\in T{\in_X}$ itself witnesses that $T\pi_1(u)\in\lambda_X(T\pi_2(u))$.

Put $v =  T\kappa(w) \in T(\in_{TX})$. Then, using naturality of $\mu$, calculate:
\[
	\mu_X(T\pi_1(v)) = \mu_X(T\pi_1(T\kappa(w))) = \mu_X(TT\pi_1(w)) = T\pi_1(\mu_{\in_X}(w)) = T\pi_1(r)=s
\]
and
\[
	T\pi_2(v) = T\pi_2(T\kappa(w)) = T\lambda_X(TT\pi_2(w)) = T\lambda_X(t)
\]
as required in~\eqref{eq:Drhs}.

For the right-to-left inclusion of~\eqref{eq:axiomd}, assume $v$ as in~\eqref{eq:Drhs}. Consider the pullback
\[\xymatrix{
Q\ar[r]^{\pi_2}\ar[d]_{\id\times\lambda_X} & TPX\ar[d]^{\lambda_X} \\
\in_{TX}\ar[r]_{\pi_2} & PTX
}\]
where
\begin{equation}\label{eq:Qdef}
	Q = \{(p,z)\mid p\in TX,\ z\in TPX,\ p\in\lambda_X(z)\}.
\end{equation}
Unfolding the definition of $\lambda_X$ in this, we obtain:
\[
	Q = \{(p,z) \mid p\in TX,\ z\in TPX,\ \exists m\in T{\in_X} \text{ s.t. } T\pi_1(m)=p,\ T\pi_2(m)=z\}.
\]
This means that we can select a function, call it $\rho: Q\to T{\in_X}$, such that
\[
	T\pi_1(\rho(p,z)) = p \qquad\text{and}\qquad T\pi_2(\rho(p,z)) = z \qquad \text{for } (p,z)\in Q.
\]

Since $T$ preserves weak pullbacks, the square
\[\xymatrix{
TQ\ar[r]^{T\pi_2}\ar[d]_{T(\id\times\lambda_X)} & TTPX\ar[d]^{T\lambda_X} \\
T{\in_{TX}}\ar[r]_{T\pi_2} & TPTX
}\]
is a weak pullback. By the properties of $v$ stated in~\eqref{eq:Drhs}, and by the definition of weak pullback, there is some $q\in TQ$ such that
\[
	T(\id\times\lambda_X)(q)=v \quad (\text{hence } T\pi_1(q)=T\pi_1(v)) \qquad\text{and} \qquad T\pi_2(q)=t.
\]
Put $r = \mu_{\in_X}(T\rho(q)) \in T{\in_X}$. Then, using naturality of $\mu$, calculate:
\[
	T\pi_1(r) = T\pi_1(\mu_{\in_X}(T\rho(q))) = \mu_X(TT\pi_1(T\rho(q))) = \mu_X(T\pi_1(q)) = \mu_X(T\pi_1(v)) = s
\]
and
\[
	T\pi_2(r) = T\pi_2(\mu_{\in_X}(T\rho(q))) = \mu_{\P X}(TT\pi_2(T\rho(q))) = \mu_{\P X}(T\pi_2(q)) = \mu_{\P X}(t)
\]
as required in~\eqref{eq:Dlhs}. This completes the proof~\eqref{eq:axiomd} and the entire theorem.
\end{proof}

Armed with $\JL$, we are ready to prove Theorem~\ref{prop:ewcm}. Let $L\subseteq TX$ be a language recognized by a finite $T$-algebra $\mathbf{A}=(A,\alpha)$ as $L=\overleftarrow{h}(C)$, where $h:TX\to \mathbf{A}$ is a homomorphism and and $C\subseteq A$.

Seeing $\JL$ as a distributive law of the monad $T$ over the functor $\PP$, we can lift the functor $\PP:\set \to \set$ to
$\widehat{\PP}:\Alg(T)\to \Alg(T)$, see \cite{beck}. In particular, as is easy to check using the axioms from Definition~\ref{def:dl},
\[
	\widehat{\PP}\mathbf{A} = (\PP A,\PP\alpha\circ\lambda_A)
\]
is a legal $T$-algebra.

For a surjective function $f:X\to Y$, define $\overline{f}:Y\to\PP X$ as the inverse image function on $Y$:
\[
	\overline{f}(y) = \{x\in X \mid f(x)=y\}.
\]

Let $g:X\to Y$ be a surjective function. Since all functors on $\set$ preserve surjectivity, $Tg:TX\to TY$ is also surjective. We will show that the direct image of $L$ along $Tg$ is recognized by the algebra $\widehat{\PP}\mathbf{A}$. To this end, consider the function $k : TY\to \PP A$ defined by the composition:
\[
k\quad=\quad	TY\overset{\overline{Tg}}{\longrightarrow}\PP TX \overset{\PP h}{\longrightarrow}\PP A
\]
and the subset $\hat{C}\subseteq \PP A$ defined by
\[\hat{C} = \{D\subseteq A \mid D\cap C\neq\emptyset\}.\]
A straightforward calculation shows that the direct image of $L$ along $Tg$ coincides with the inverse image $ \overleftarrow{k}(\hat{C})$.

To finish the proof it is enough to show that $k$ is a $T$-algebra morphism.
That is, we need to show that the outer shape of the following diagram commutes:
\[\xymatrix{
TTY\ar[r]^{T\overline{Tg}}\ar[dd]_{\mu_Y}\ar[rd]_{\overline{TTg}}\ar@{}[rrd]|(.3){(1)} & T\PP TX\ar[d]^{\JL_{TX}}\ar[r]^{T\PP h}\ar@{}[rd]^(.6){(3)} & T\PP A\ar[d]^{\lambda_A}\\
& \PP TTX\ar[d]^{\PP\mu_X}\ar[r]^{\PP Th}\ar@{}[rd]^(.6){(4)} & \PP TA\ar[d]^{\PP\alpha}\\
TY\ar[r]_{\overline{Tg}}\ar@{}[ur]|-{(2)} & \PP TX\ar[r]_{\PP h} & \PP A
}\]
Part (3) commutes by naturality of $\lambda$ and part (4) commutes since $h$ is a $T$-algebra homomorphism. We now show the commutativity of (1) and (2).

For (2), given $s \in TTY$, we have:
\begin{align*}
	\PP \mu_X(\overline{TTg}(s))&=\{\mu_X(r)\mid r\in TTX\text{ and } TTg(r)=s\},\\
	\overline{Tg}(\mu_Y(s))&=\{t\in TX\mid \mu_Y(s)=Tg(t)\}.
\end{align*}
To prove $\PP \mu_X(\overline{TTg}(s))=\overline{Tg}(\mu_Y(s))$ we consider both inclusions.

For the left-to-right inclusion, put $t:=\mu_X(r)$ and
use naturality of $\mu$. For the right-to-left inclusion, use the fact that, since $T$ is weakly epi-cartesian,
\[\xymatrix{
	TTX\ar[r]^{TTg}\ar[d]_{\mu_X} & TTY\ar[d]^{\mu_Y} \\
	TX\ar[r]_{Tg} & TY
}\]
is a weak pullback.

To prove that (1) commutes, we will use the following simple fact:

\begin{lemma}\label{lem:small}
	Let $f:X\to Y$ be a surjective function. Then:
	\begin{enumerate}[i)]
		\item $\eta^{\P}_X\subseteq \bar{f}\circ f$ and $\bar{f}$ is the minimum function from $Y$ to $\PP X$ with this property,
		\item $\PP f\circ \bar{f}\subseteq \eta^{\P}_Y$ and $\bar{f}$ is the maximum function from $Y$ to $\PP X$ with this property. \qed
	\end{enumerate}
\end{lemma}

Using this together with Proposition~\ref{prop:Jacmon} and the axiom (a) from Definition~\ref{def:dl} for $\JL$, we obtain:
\begin{align*}
	\eta^{\P}_{TX}\subseteq \overline{Tg}\circ Tg\ &\Rightarrow\ \eta^{\P}_{TTX}=\JL_{TX}\circ T\eta^{\P}_{TX}\subseteq \JL_{TX}\circ
			T\overline{Tg}\circ TTg\\
			&\Rightarrow\ \overline{TTg}\subseteq \JL_{TX}\circ T\overline{Tg}
\end{align*}
Similarly, we obtain that:
\begin{align*}
	\PP Tg\circ \overline{Tg}\subseteq \eta^{\P}_{TY}\ &\Rightarrow\ \JL_{TY}\circ T\PP Tg\circ T\overline{Tg}\subseteq
			\JL_{TY}\circ T\eta^{\P}_{TY} = \eta^{\P}_{TTY}\\
			&\Rightarrow\ \PP TTg\circ \JL_{TX}\circ T\overline{Tg}\subseteq \eta^{\P}_{TTY}\\
			&\Rightarrow\ \JL_{TX}\circ T\overline{Tg}\subseteq \overline{TTg}
\end{align*}
which completes the proof that (1) commutes, and the proof of the entire theorem.
\qed

\section{Other examples}\label{sec:sporadic}

In this section, we show a few more cases where $\rec_{T}$ is closed under direct images along surjective letter-to-letter homomorphisms. These cases do not satisfy any of the sufficient conditions presented so far. They are also quite varied, with no clear pattern emerging:
\begin{itemize}
\item the reader monad in Section~\ref{sec:reader} is infinitary, but its finite algebras are rather restricted, and preservation under direct images is proved by a compactness argument,
\item for the free lattice monad (Section~\ref{sec:freelattice}), preservation follows from a convexity-based argument,
\item the monad in Section~\ref{sec:fgfgg} is rather peculiar in that almost no language is recognizable for it, so preservation of recognizable languages under direct images holds for trivial reasons,
\item for the monad in Section~\ref{sec:xyyz-xyz}, a construction similar to that described in Section~\ref{sec:suffcond-jacobs} works, but with the powerset replaced by a more complex powerset-squared construction.
\end{itemize}

\subsection{The reader monad}\label{sec:reader}

\begin{monad}\label{mon:reader}
Let $\Sigma^{\omega}$ denote the set of $\omega$-sequences of letters from $\Sigma$. The functor $T=(-)^\omega$ acts on functions in the expected way:
\[
	f^+(a_0a_1a_2\cdots) = (f(a_0))(f(a_1))(f(a_2))\cdots.
\]
This is an infinitary monad, with the unit mapping a letter $a\in\Sigma$ to the constant sequence $aaa\cdots$, and the multiplication defined by the diagonal function:
\[
	\mu_\Sigma(w_0w_1w_2\cdots) = w_{0,0}w_{1,1}w_{2,2}\cdots
\]
This is a particular instance of the Haskell ``reader'' monad.
\end{monad}

We will prove that recognizable $T$-languages are closed under taking direct images along surjective letter-to-letter homomorphisms. To fix the notation, we choose an arbitrary surjective function $f:\Sigma\to \Gamma$ for $\Sigma,\Gamma$ finite, and any language $L\subseteq T\Sigma$ recognizable by a homorphism $h:T\Sigma\to \mathbf{A}$ to a finite $T$-algebra $\mathbf{A}=(A,\alpha)$. This means that for some subset $S\subseteq A$:
\[
	t\in L \iff h(t)\in S.
\]
Assuming all this, we shall prove that the direct image of $L$ along $Tf$ is recognizable by a homomorphism from $T\Gamma$ to a finite $T$-algebra.

Without loss of generality we may assume that $L$ is recognizable by a point in $A$, i.e. that $S=\{a\}$
for some $a\in A$. This is because both taking inverse images and taking direct images commutes with unions, and recognizable languages are closed under finite unions. Such languages have a rather rigid ``rectangular'' structure:
\begin{lemma}\label{lem:rectangular}
If $L\subseteq T\Sigma$ is recognizable by a point then
\begin{equation}\label{eq:rectangular}
	L = Z_1\times Z_2\times Z_3 \times\cdots
\end{equation}
for some sequence of subsets $Z_1,Z_2,\ldots\subseteq \Sigma$.
\end{lemma}
\begin{proof}
Given a language $L$ recognizable by $a\in A$ for an algebra $\alpha:TA\to A$, define
\[
	Z_n = \{x \in \Sigma \mid \exists v\in L. v_n=x\}.
\]
for each $n\in\mathbb{N}$. Then the left-to-right inclusion in~\eqref{eq:rectangular} is obvious. For the right-to-left inclusion, consider a word $w\in Z_1\times Z_2\times\cdots$; this means that for every $n\in\mathbb{N}$ there is a word $v_n\in L$ such that $v_{n,n}=w_n$. Then $w=\mu_\Sigma(v_1v_2\cdots)$, so:
\[
	h(w) = h(\mu_\Sigma(v_1v_2\cdots)) = \alpha(Th(v_1v_2\cdots)) = \alpha(aaa\cdots) = a
\]
so $w\in L$.
\end{proof}

For any $\Sigma$, the set $T\Sigma$ is equipped with the product topology, whose basis is the family of all {\em open balls} defined by
\[
	{\cal B}_I(w) = \{v\in T\Sigma \mid \forall n\in I.\ v_n=w_n\}
\]
for $w\in T\Sigma$ and finite $I\subseteq\mathbb{N}$. With respect to this topology:

\begin{lemma}\label{lem:reader-closed}
Every $L$ recognizable by a point is closed.
\end{lemma}
\begin{proof}
From Lemma~\ref{lem:rectangular} it easily follows that $T\Sigma\setminus L$ is open. Indeed, for $L$ as in~\eqref{eq:rectangular}, take any $w\not\in L$. Then there is an $n\in\mathbb{N}$ such that $w_n\not\in Z_n$, and then ${\cal B}_{\{n\}}(w)$ is disjoint from $L$.
\end{proof}

\begin{lemma}\label{lem:reader-open}
Every $L$ recognizable by a point in a finite algebra $A$ is open.
\end{lemma}
\begin{proof}
Let $L=\overleftarrow{f}\{a\}$ for some $a\in A$. Then
\[
	T\Sigma\setminus L = \bigcup_{b\in A\setminus\{a\}}\overleftarrow{f}\{b\},
\]
each $\overleftarrow{f}\{b\}$ is closed by Lemma~\ref{lem:reader-closed}, and their union is closed since a finite union of closed sets is closed.
\end{proof}

We can now prove a stronger version of Lemma~\ref{lem:rectangular}. Call a language $L\in T\Sigma$ {\em finitely rectangular} if
\begin{equation}\label{eq:fin-rectangular}
		L = Z_1\times Z_2\times Z_3 \times\cdots
\end{equation}
for some sequence of subsets $Z_1,Z_2,\ldots\subseteq X$ such that $Z_n=\Sigma$ for all except finitely many $n$.

\begin{lemma}\label{lem:rectangular2}
Every nonempty $L$ recognizable by a point in a finite algebra is finitely rectangular.
\end{lemma}
\begin{proof}
By Lemma~\ref{lem:rectangular}, $L$ is of the form
\[
		L = Z_1\times Z_2\times Z_3 \times\cdots;
\]
we need to prove that $Z_n=\Sigma$ for almost all $n$.

Since $\Sigma$ is a finite set, the space $T\Sigma$ is compact. Its subspace $L$ is closed by Lemma~\ref{lem:reader-closed}, so it is also compact. Consider the family of all those open balls ${\cal B}_I(w)$ that are contained in $L$. Since $L$ is open by Lemma~\ref{lem:reader-open}, for each $w\in L$ there is some $I$ such that ${\cal B}_I(w)$ belongs to this family, so the family is an open cover of $L$. By compactness, it contains a finite subcover
\begin{equation}\label{eq:subcover}
	{\cal B}_{I_1}(w_1),\ldots,{\cal B}_{I_k}(w_k).
\end{equation}
Let $I=I_1\cup\cdots\cup I_k$. Fix any $w\in L$ and take any $v\in T\Sigma$ such that $w_n=v_n$ for all $n\in I$. Since $L$ is covered by~\eqref{eq:subcover} we have $w\in{\cal B}_{I_i}(w_i)$ for some $i$. Since $I_i\subseteq I$ also $v\in{\cal B}_{I_i}(w_i)$, so $v\in L$. Since we did not constrain $v$ on the coordinates outside of $I$, we have that $Z_n=\Sigma$ for $n\not\in I$.
\end{proof}

Finitely rectangular languages are closed under direct images of surjective homomorphisms:

\begin{lemma}
For a surjective function $f:\Sigma\to \Gamma$, if $L\subseteq T\Sigma$ is finitely rectangular then so is the direct image $\overrightarrow{Tf}(L)\subseteq T\Gamma$.
\end{lemma}
\begin{proof}
Obviously
\[
	\overrightarrow{Tf}(Z_1\times Z_2\times\cdots) = \overrightarrow{f}(Z_1)\times \overrightarrow{f}(Z_2)\times\cdots,
\]
and $\overrightarrow{f}(\Sigma)=\Gamma$ since $f$ is surjective.
\end{proof}

The final piece of the puzzle is:

\begin{lemma}
Every finitely rectangular language is recognizable by a finite algebra.
\end{lemma}
\begin{proof}
Let $L\subseteq T\Sigma$ be as in~\eqref{eq:fin-rectangular}, and let $I\subseteq\mathbb{N}$ be the (finite) set of those $n$ where $Z_n\neq \Sigma$. Define $A = \Sigma^I$, and let $\alpha:TA\to A$ be defined by:
\[
	\alpha(a_1a_2\cdots)(n) = a_n(n) \qquad \text{for }n\in I.
\]
We need to check that $\alpha$ is a $T$-algebra. For the unit axiom, given $a\in A$ and $n\in I$, calculate:
\[
	\alpha(\eta_A(a))(n) = \alpha(aaa\cdots)(n) = a(n).
\]
For the multiplication axiom, given $\mathbf{v}=v_1v_2v_3\cdots\in TTA$ and $n\in I$, calculate:
\[
	\alpha(\mu_A(\mathbf{v}))(n) = \alpha(v_{1,1}v_{2,2}\cdots)(n) = v_{n,n}(n)
\]
and
\[
	\alpha(T\alpha(\mathbf{v}))(n) = \alpha(\alpha(v_1),\alpha(v_2),\ldots)(n) = \alpha(v_n)(n) = v_{n,n}(n)
\]
so both expressions are equal as required, hence $\alpha$ is a $T$-algebra structure on $A$.

Define a homomorphism $h:T\Sigma\to A$ as a unique homomorphic extension of the function $h_0:\Sigma\to A$ that maps every $x\in \Sigma$ to the function in $\Sigma^I$ constant at $x$. Explicitly, this is defined by:
\[
	h(w)(n) = \alpha(Th_0(w))(n) = w_n \qquad \text{for } w\in T\Sigma, n\in I.
\]
Recalling~\eqref{eq:fin-rectangular}, define $S\subseteq A$ by:
\[
	S = \{g\in \Sigma^I \mid \forall n\in I.\ g(n)\in Z_n\}.
\]
Then, for $w\in T\Sigma$, calculate:
\[
	w\in L \iff \forall n\in I. w_n\in Z_n \iff h(w)\in S,
\]
hence $L$ is recognized by $S\subseteq A$ along the homomorphism $h$.
\end{proof}

\subsection{The free lattice monad}\label{sec:freelattice}

\begin{monad}\label{mon:lattice}
The free lattice monad $\FL$, whose algebras are lattices, has a well-known equational presentation over a signature with two binary symbols $\lor$ and $\land$. It consists of equations:
\[\begin{array}{ccccc}
x\lor y = y\lor x & \qquad & x \lor (y\lor z) = (x \lor y)\lor z & \qquad & x \lor (x \land y) = x \\
x\land y = y\land x & \qquad & x \land (y\land z) = (x \land y)\land z & \qquad & x \land (x \lor y) = x
\end{array}\]
Lattices can be usefully regarded as partial orders, defined by:
\[
	x \leq y \qquad\text{iff} \qquad x\lor y = y;
\]
$\lor$ and $\land$ then become respectively join and meet operations.
\end{monad}

We will prove that recognizable $\FL$-languages are closed under taking direct images along surjective letter-to-letter homomorphisms. So, to fix the notation, we choose an arbitrary surjective function $f:\Sigma\to \Gamma$ for $\Sigma,\Gamma$ finite, and any language $L\subseteq\FL(\Sigma)$ recognizable by a homorphism $h:\FL(\Sigma)\to A$ to a finite lattice $A$. This means that for some subset $S\subseteq A$:
\[
	t\in L \iff h(t)\in S.
\]
Assuming all this, we shall prove that $\overrightarrow{Tf}(L)$, the direct image of $L$ along $Tf$, is recognizable by a homomorphism from
$\FL(\Gamma)$ to a finite lattice.

First let us consider the case where $L$ is upwards-closed, i.e., where
\[
	s\in L \quad\text{and}\quad s\leq t \quad\text{implies}\quad t\in L.
\]
Consider a homomorphism $k:\FL(\Gamma)\to\FL(\Sigma)$ defined as the unique extension of the map $k_0:\Gamma\to\FL(\Sigma)$ defined by:
\[
	k_0(y) = \bigvee\{x \mid f(x)=y\}.
\]
Note that the above join is nonempty since $f$ is surjective. In particular, we have $f(k_0(y))=y$ for all $y\in \Gamma$ and, by definition, $k_0(f(x))\geq x$ for all $x\in \Sigma$. Since $Tf$ and $k$ preserve meets and joins by definition, this extends to:
\begin{equation}\label{eq:FL-galois}
	Tf(k(t))=t \qquad\text{and}\qquad k(Tf(s))\geq s
\end{equation}
for all $t\in\FL(\Gamma)$ and $s\in\FL(\Sigma)$.

Then it is easy to see that
\begin{equation}\label{eq:im-vs-inv-im}
\overrightarrow{Tf}(L)=\overleftarrow{k}(L).
\end{equation}
Indeed, for the right-to-left inclusion, for $k(t)\in L$ one has
\[
	t \stackrel{\eqref{eq:FL-galois}}{=} Tf(k(t)) \in \overrightarrow{Tf}(L).
\]
For the left-to-right inclusion, take any $t\in\FL(\Gamma)$ such that $t=Tf(s)$ for some $s\in L$. Then
\[
	k(t) = k(Tf(s)) \stackrel{\eqref{eq:FL-galois}}{\geq} s
\]
and, since $L$ is upwards-closed, $k(t)\in L$.

The equation~\eqref{eq:im-vs-inv-im} implies that $\overrightarrow{Tf}(L)$ is the inverse image of $S\subseteq A$ along the composite homomorphism $h\circ k:\FL(Y)\to A$, and so it is recognizable.

An analogous argument works if $L$ is downwards-closed, with $k$ defined using a meet rather than a join.

Now consider an arbitrary (i.e.~not necessarily upwards- or downwards-closed) language $L$ recognizable by a homomorphism $h:\FL(\Sigma)\to A$ to a finite lattice $A$. Without loss of generality we may assume that $L$ is recognizable by a point in $A$, i.e., that
\[
L = \overleftarrow{h}\{a\}
\]
for some $a\in A$. This is because both taking inverse images and taking direct images commutes with unions, and recognizable languages are closed under finite unions.

Let $L{\uparrow}$ and $L{\downarrow}$ denote the upwards-closure and the downwards-closure of $L$, respectively. It is not difficult to check that upwards closure commutes with inverse images, in particular:
\[
	\overleftarrow{h}(a{\uparrow}) = L{\uparrow}.
\]
Indeed, this amounts to requiring that, for every $t\in\FL(\Sigma)$,
\[
	h(t)\geq a \iff \exists s\leq t. h(s)=a.
\]
The right-to-left implication is immediate since $h$ is monotone. For the left-to-right implication, since $h:\FL(\Sigma)\to A$ is surjective, there is some $r\in\FL(\Sigma)$ such that $h(r)=a$. Put $s=t\land r$. Then obviously $s\leq t$, and
\[
	h(s) = h(t)\land h(r) = h(t)\land a = a
\]
as required.

An analogous argument works for downward closures, i.e.:
\[
	\overleftarrow{h}(a{\downarrow}) = L{\downarrow}.
\]
This means that both $L{\uparrow}$ and $L{\downarrow}$ are recognizable (by $h$), hence by the previous argument, their direct images along $Tf$ are also recognizable. Since recognizable languages are closed under intersection, it is now enough to prove that
\begin{equation}\label{eq:uparr-downarr}
\overrightarrow{Tf}(L{\uparrow})\cap \overrightarrow{Tf}(L{\downarrow}) = \overrightarrow{Tf}(L).
\end{equation}
To this end, first notice that, since $L$ is recognized by a point, it is {\em convex}, i.e.,
\[
	s\leq t\leq r \quad \text{and}\quad s,r\in L \quad\text{implies}\quad t\in L.
\]
This implies that
\[
	L{\uparrow}\cap L{\downarrow} = L;
\]
indeed the right-to-left inclusion is trivial, and the left-to-right inclusion easily follows from the convexity of $L$. This does not immediately imply~\eqref{eq:uparr-downarr}, as taking direct images does not commute with intersections in general. In this case, however, it does. To see this, first notice that $L{\uparrow}$ is a {\em filter}, i.e., it is upward-closed and closed under finite intersections. Similarly, $L{\downarrow}$ is an {\em ideal}, i.e., it is downward-closed and closed under finite unions. Now all we need is the following lemma, which holds for arbitrary lattice homomorphisms:
\begin{lemma}
For any lattice homomorphism $g:H\to K$ and any filter $U\subseteq H$ and ideal $D\subseteq H$ such that $U\cap D\neq\emptyset$:
\[
	\overrightarrow{g}(U\cap D) = \overrightarrow{g}(U)\cap \overrightarrow{g}(D).
\]
\end{lemma}
\begin{proof}
The left-to-right inclusion is obvious. For the right-to-left inclusion, take any $t\in K$ such that
\[
	\exists u\in U. h(u)=t \qquad\text{and}\qquad \exists d\in D.h(d)=t.
\]
Pick any $x\in U\cap D$ (it exists by our assumptions), and put:
\[
	y = d\lor(x\land u).
\]
Since $x\in D$ and $D$ is downward-closed, also $x\land u\in D$ and (since $D$ is closed under unions) $y\in D$.
On the other hand, since $x\in U$ and $U$ is closed under intersections, also $x\land u\in U$ and (since $U$ is upward-closed) $y\in U$. So $y\in D\cap U$, and it is enough to calculate:
\[
	h(y) = h(d) \lor (h(x)\land h(u)) = t\lor(h(x)\land t) = t.
\]
\end{proof}

Using this lemma for $g=Tf$, $U=L{\uparrow}$ and $D=L{\downarrow}$, we obtain~\eqref{eq:uparr-downarr} which directly implies that $L$ is recognizable.

\subsection{A non--trivial monad whose finite algebras are trivial}\label{sec:fgfgg}

\begin{monad}\label{mon:fgfgg}
Let $T$ be the monad associated to the signature $\{f,g\}$, where $f$ and $g$ are unary operations, subject to axioms:
\[
	fgfgg(x)=x\text{ \ \ and \ \ }fgffgg(x)=fgffgg(y).
\]
\end{monad}

It is shown that this equational theory has no non-trivial finite models, that is, every non--empty finite model has only one element~\cite[Section 2]{Burris1971}. Indeed, if $(A,f,g)$ is a finite algebra, from the equation $fgfgg(x)=x$ we obtain that $f$ is  a surjection and $g$ is a injection, which implies that $f$ and $g$ are bijections since $A$ is finite and $f:A\to A$ and $g:A\to A$. Hence, from the equation $fgffgg(x)=fgffgg(y)$, we conclude that $x=y$ for every $x,y\in A$. Therefore, $\rec_{T}$ is trivially closed under direct images along surjective letter--to--letter homomorphisms.

\begin{rem}
Even though the monad has the trivial algebra as the only finite model, the monad has infinite models. One such model is the set of the positive integers with $g(n)=2^n$, $f(2^{2^n})=3^n$, $f(2^{3^n})=n$, and $f(m)=1$ otherwise, as shown in~\cite[Section 2]{Burris1971}. Also, it is worth mentioning that given a finite set of identities over a given finitary signature, it is undecidable if there exists a non-trivial finite model that satisfies the given identities~\cite{mckenzie}. That is, it is undecidable in general if the class $\rec_{T}$ of recognizable languages is non-trivial. This means that describing the class $\rec_{T}$ or even finding a non--trivial element in $\rec_{T}$ could be undecidable.
\end{rem}

\subsection{The $x(y(yz))=x(yz)$ monad}\label{sec:xyyz-xyz}

\begin{monad}\label{mon:xyyz-xyz}
Consider a monad $T$ presented by an equational theory with a single binary symbol $\mydot$ and with a single equation:
\begin{equation}\label{eq:xyyz-xyz}
\vcenter{\xymatrix@=.6em{
& \mydot\ar@{-}[ld]\ar@{-}[rd] \\
x & & \mydot\ar@{-}[ld]\ar@{-}[rd] \\
& y & & \mydot\ar@{-}[ld]\ar@{-}[rd] \\
& & y & & z
}}
\quad=\quad
\vcenter{\xymatrix@=.6em{
& \mydot\ar@{-}[ld]\ar@{-}[rd] \\
x & & \mydot\ar@{-}[ld]\ar@{-}[rd] \\
& y & & z
}}.
\end{equation}
\end{monad}

This equation is similar to the one that defined Monad~\ref{mon:xxy-xy}, but the resulting monad has rather different properties. For one, its multiplication is not weakly epi-cartesian. To see this, consider
\[
	\Sigma = \{a_1,a_2,b,c\} \qquad \Gamma = \{a,b,c\}
\]
and $f:\Sigma\to \Gamma$ defined by $f(a_1)=f(a_2)=a$, $f(b)=b$ and $f(c)=c$. Define $t\in TT\Gamma$ by:
\[
	t = \vcenter{\xymatrix@=.6em{& \mydot\ar@{-}[ld]\ar@{-}[rd] \\
b & & t'}}
\qquad\mbox{where}\qquad
	t' = \vcenter{\xymatrix@=.6em{& \mydot\ar@{-}[ld]\ar@{-}[rd] \\
a & & c}} \in T\Gamma \quad \mbox{ and } b \mbox{ is seen as } \eta_\Gamma(b)\in T\Gamma.
\]
Then, obviously:
\[
	\mu_Y(t) = \vcenter{\xymatrix@=.6em{
& \mydot\ar@{-}[ld]\ar@{-}[rd] \\
b & & \mydot\ar@{-}[ld]\ar@{-}[rd] \\
& a & & c.
}}
\]
Now let $r\in T\Sigma$ be defined by:
\[
	r = \vcenter{\xymatrix@=.6em{
& \mydot\ar@{-}[ld]\ar@{-}[rd] \\
b & & \mydot\ar@{-}[ld]\ar@{-}[rd] \\
& a_1 & & \mydot\ar@{-}[ld]\ar@{-}[rd] \\
& & a_2 & & c.
}}
\]
It is easy to see that $Tf(r)=\mu_Y(t)$. However, there is no $p\in TT\Sigma$ such that $TTf(p)=t$ and $\mu_\Sigma(p)=r$. As a result, the naturality square
\[\xymatrix{
TT\Sigma\ar[r]^{\mu_\Sigma}\ar[d]_{TTf} & T\Sigma\ar[d]^{Tf} \\
TT\Gamma\ar[r]_{\mu_\Gamma} & T\Gamma
}\]
is not a weak pullback.

In spite of this, recognizable $T$-languages are closed under taking direct images along surjective letter-to-letter homomorphisms.

A $T$-algebra is simply a set equipped with a binary operation which satisfies~\eqref{eq:xyyz-xyz}. Assume finite alphabets $\Sigma$ and $\Gamma$, a surjective function $f:\Sigma\to\Gamma$, and a language $L\subseteq T\Sigma$ that is recognized by a subset $S\subseteq A$ of a finite $T$-algebra $A$; in other words, $L=\overleftarrow{h}(S)$ for some $T$-algebra homomorphism $h:T\Sigma\to A$. We will show that $\overrightarrow{Tf}(L)$ is recognized by a finite $T$-algebra.

To this end, define a $T$-algebra
\[
	\bar{A} = \PP A\times \PP A.
\]
The intuition is that a tree $t\in T\Gamma$ will be mapped to a value $(\alpha_L,\alpha_R)\in \bar{A}$ such that $\alpha_L$ contains those values in $A$ that can be attained by trees in $T\Sigma$ which are mapped to $t$ and which are left sons of their parents. Similarly, $\alpha_R$ will stores values for trees that are right sons of their parents.

Formally, a binary operation $\mydot$ is defined by:
\begin{align*}
	(\alpha_L,\alpha_R)\mydot(\beta_L,\beta_R) =&\ \big(\{a\mydot b \mid a\in\alpha_L, b\in\beta_R\}, \\
	&\ \ \{a_1\mydot(a_2\mydot(\cdots(a_n\mydot b)\cdot\cdot)) \mid n\geq 1, a_i\in\alpha_L, b\in\beta_R\}\big),
\end{align*}
where $\mydot$ on the right-hand side of the definition denotes the operation in $A$. It is not difficult to check that this defines a $T$-algebra, i.e., that the equation
\begin{equation}\label{eq:xyyz-xyz1}
\vcenter{\xymatrix@=.6em{
& \mydot\ar@{-}[ld]\ar@{-}[rd] \\
(\alpha_L,\alpha_R) & & \mydot\ar@{-}[ld]\ar@{-}[rd] \\
& (\beta_L,\beta_R) & & \mydot\ar@{-}[ld]\ar@{-}[rd] \\
& & (\beta_L,\beta_R) & & (\gamma_L,\gamma_R)
}}
\quad=\quad
\vcenter{\xymatrix@=.6em{
& \mydot\ar@{-}[ld]\ar@{-}[rd] \\
(\alpha_L,\alpha_R) & & \mydot\ar@{-}[ld]\ar@{-}[rd] \\
& (\beta_L,\beta_R) & & (\gamma_L,\gamma_R)
}}
\end{equation}
holds for all $\alpha_L,\alpha_R,\beta_L,\beta_R,\gamma_L,\gamma_R\subseteq A$. The (slightly) more involved case is the right-to-left inclusion, where one needs to use the equation~\eqref{eq:xyyz-xyz} for $A$. For the left-to-right inclusion even that is not needed.

A homomorphism $\bar{h}:T\Gamma\to \bar{A}$ is the unique extension of the function that maps every letter $a\in\Gamma$ to
\[
	\left(\overrightarrow{h\circ\eta_{\Sigma}}(\overleftarrow{f}(a)),\overrightarrow{h\circ\eta_{\Sigma}}(\overleftarrow{f}(a))\right)
\]
(note that $h\circ\eta_\Sigma:\Sigma\to A$ is the interpretation of single letters from $\Sigma$ in the algebra A).

Now define a subset $\bar{S}\subseteq \bar{A}$ by:
\[
	\bar{S} = \{(\alpha_L,\alpha_R) \mid \alpha_L\cap S \neq\emptyset\}.
\]
This subset recognizes the direct image $\overrightarrow{Tf}(L)$. More explicitly, for any tree $t\in T\Gamma$, one has $\bar{h}(t)\in\bar{S}$ if and only if $t=Tf(r)$ for some $r\in T\Sigma$ such that $h(r)\in S$. Both implications are proved by a straightforward induction on the size of $t$.

\section{Counterexamples}\label{sec:counterexamples}

In this section, we illustrate cases where direct images of recognizable languages along surjective letter-to-letter homomorphisms are not recognizable. One such case, Monad~\ref{mon:x3-x2}, was shown in Section~\ref{sec:monadic:mso}. Here we provide three more.

\subsection{The marked words monad}\label{sec:markedwords}

Recall Monad~\ref{mon:semigroup}, i.e. the free semigroup monad $(-)^{+}$, and Monad~\ref{mon:bag}, i.e., the bag monad $\Bag$.
For $w\in \Sigma^+$ and  $\beta\in\Bag \Sigma$, we write $w\sqsupseteq \beta$ if each $x\in \Sigma$ occurs in $w$ at least $\beta(x)$ times.

\begin{monad}\label{mon:marked}
Define a functor $T$ on $\set$ by:
\[
	T\Sigma = \{(w,\beta)\in \Sigma^+\times\Bag \Sigma \mid w\sqsupseteq \beta\}
\]
with the action on fuctions defined componentwise.
It is easy to check that this is well-defined, i.e., that
\[
	w\sqsupseteq\beta \text{ implies } f^+(w)\sqsupseteq \Bag f(\beta).
\]
$T$ carries a monad structure where both unit and multiplication are componentwise inherited from the monads $(-)^+$ and $\Bag$.
Again, it is straightforward to check that this is well defined.

To get some intuition, an element of $T\Sigma$ can be understood as a finite, nonempty word over the alphabet $\Sigma$ with some positions marked, except that it is not specified exactly which positions those are; we only know how many positions labeled with every letter are marked. So, for example, markings
\[
	\underline{a}\underline{a}bca\underline{b} \qquad \underline{a}a\underline{b}c\underline{a}b \qquad \text{and} \qquad a\underline{a}\underline{b}c\underline{a}b
\]
denote the same element of $T\{a,b,c\}$. This intuition is useful because it allows for a simple presentation of the monad structure: the unit of $a\in \Sigma$ is unambiguosly presented as $\underline{a}$
and multiplication can be illustrated as, for example:
\begin{gather*}
	\underline{(\underline{a}bac)}(a\underline{bb})(ab\underline{a}c)\underline{(b\underline{a}c)} \qquad\mapsto\qquad
	\underline{a}bacabbabacb\underline{a}c.
\end{gather*}

This monad has a simple equational presentation over an algebraic signature that consists of one binary symbol $\mydot$ and one unary symbol $\circ$. Using the terminology of marked words, the intuition will be that $\mydot$ is an operation that concatenates two words, and $\circ$ {\em erases} the marking from every letter in a word. The equations are as follows:
\[
\vcenter{\xymatrix@=.6em{
& & \mydot\ar@{-}[ld]\ar@{-}[rd] \\
& \mydot\ar@{-}[ld]\ar@{-}[rd] & & z \\
x & & y
}}
\quad=\quad
\vcenter{\xymatrix@=.6em{
& \mydot\ar@{-}[ld]\ar@{-}[rd] \\
x & & \mydot\ar@{-}[ld]\ar@{-}[rd] \\
& y & & z
}}
\quad,\quad
\vcenter{\xymatrix@=.6em{
& \circ\ar@{-}[d]\\
& \mydot\ar@{-}[ld]\ar@{-}[rd] \\
x & & y
}}
\quad=\quad
\vcenter{\xymatrix@=.6em{
& \mydot\ar@{-}[ld]\ar@{-}[rd] \\
\circ\ar@{-}[d] & & \circ\ar@{-}[d] \\
x & & y
}}
\quad,\quad
\vcenter{\xymatrix@=.6em{
\circ\ar@{-}[d]\\
\circ\ar@{-}[d]\\
x
}}
\quad=\quad
\vcenter{\xymatrix@=.6em{
\circ\ar@{-}[d]\\
x
}}
\quad,
\]
\[
\vcenter{\xymatrix@=.6em{
& & \mydot\ar@{-}[ld]\ar@{-}[rd] \\
& \mydot\ar@{-}[ld]\ar@{-}[rd] & & \circ\ar@{-}[d] \\
x & & y & x
}}
\quad=\quad
\vcenter{\xymatrix@=.6em{
& \mydot\ar@{-}[ld]\ar@{-}[rd] \\
\circ\ar@{-}[d] & & \mydot\ar@{-}[ld]\ar@{-}[rd] \\
x & y & & x
}}
\]
The first three equations ensure that every element of $T\Sigma$ can be presented as a marked word over $\Sigma$, and the last equation characterizes the equivalence of marked words described above.

This completes our exhibition of the monad $T$.
\end{monad}

Now, for the alphabet $\Sigma=\{a,b\}$, consider the language $L\subseteq T\Sigma$ defined by:
\[
	L = \{\left((ab)^n,[a\mapsto n, b\mapsto 0]\right) \mid n\in \mathbb{N}\}.
\]
Using the presentation by marked words, $L$ consists of all words of the form
\[
	\underline{a}b\underline{a}b\cdots\underline{a}b.
\]

We claim that $L$ is recognizable by a finite $T$-algebra. Intuitively, this is because a word $w\in L$ cannot be obtained by concatenating or erasing marks from any other marked words, except by concatenating subwords of $w$ or perhaps erasing a mark from a single letter $\underline{b}$.

The algebra to recognize $L$ has seven values:
\[
	A = \{\StAA,\StAB,\StBA,\StBB, \StB, \StuB, \bot\}.
\]
Intuitively, the value $\StBA$ represents those sub-words of words from $L$ that begin with $b$ and end with $\underline{a}$; analogously for $\StAA$, $\StAB$ and $\StBB$, except that the latter value represents only words other than the singleton $b$. That word is represented by a separate value $\StB$. The value $\StuB$ represents a singleton marked letter $\underline{b}$. Finally, $\bot$ is an error value, representing those words that cannot be completed to a word in $L$ by erasing marks from all letters or by concatenating with other words.

Relying on the equational presentation of $T$, to define a $T$-algebra on $A$ it is enough to define operations $\mydot$ and $\circ$ on it. The multiplication table for $x\mydot y$ is:
\[\begin{array}{c|c|c|c|c|c|c|c|}
\raisebox{-1ex}[1ex][1.5ex]{$x$}\mydot \raisebox{1ex}[1.5ex][1.5ex]{$y$}& \StAA & \StAB & \StBA & \StBB & \StB & \StuB & \bot \\
\hline
\StAA & \bot & \bot & \StAA & \StAB & \StAB & \bot & \bot \\
\hline
\StAB & \StAA & \StAB & \bot & \bot & \bot & \bot & \bot  \\
\hline
\StBA & \bot & \bot & \StBA & \StBB & \StBB & \bot & \bot \\
\hline
\StBB & \StBA & \StBB & \bot & \bot & \bot & \bot & \bot \\
\hline
\StB & \StBA & \StBB & \bot & \bot & \bot & \bot & \bot \\
\hline
\StuB  & \bot & \bot & \bot & \bot & \bot & \bot & \bot \\
\hline
\bot & \bot & \bot & \bot & \bot & \bot & \bot & \bot \\
\hline
\end{array}
\]
Note that $\StB$ and $\StBB$ behave in the same way in this table, and similarly for $\StuB$ and $\bot$. However, these pairs of values are distinguished by the operation $\circ$, which is defined by:
\[
\circ(\StB) = \circ(\StuB) = \StB \qquad\text{and} \qquad {\circ(x)}=\bot \quad\text{for } x\not\in\{\StB,\StuB\}.
\]
It is easy to check that the operations defined this way satisfy the equational theory of $T$ given above, therefore they make $A$ a $T$-algebra.

Since $T\Sigma$ is a free $T$-algebra on $\Sigma$, any function from $\Sigma$ to the set $A$ extends uniquely to a homomorphism from $T\Sigma$ to the algebra $A$. Let $h$ be the homomorphism that extends the function mapping $a$ to $\StAA$ and $b$ to $\StuB$. Then we have, for every $w\in T\Sigma$:
\[
	w\in L \iff h(w)=\StAB
\]
so $L$ is recognized by $A$.

Now consider an alphabet $Y=\{c\}$ and  the unique function $f:\Sigma\to \Gamma$, which maps both $a$ and $b$ to $c$.
The direct image of $L$ under $Tf$ is:
\begin{equation}\label{eq:c2n}
	\left\{\left(c^{2n},[c\mapsto n]\right) \mid n\in \mathbb{N}\right\}.
\end{equation}
Using the presentation by marked words, this language consists of all words of the form
\[
	\underline{c}c\underline{c}c\cdots\underline{c}c = \overbrace{\underline{cc\cdots cc}}^{n \text{ times}}\underbrace{cc\cdots cc}_{n \text{ times}}.
\]
This language is not recognizable by any finite $T$-algebra. For assume any homomorphism $h$ from $T\Gamma$ to some finite algebra $A$, and consider $w_1,w_2,w_3,\ldots\in T\Gamma$ defined by:
\[
	w_n = (c^n,[c\mapsto n]) = \underline{c}^n.
\]
Since $A$ is finite, there are some $n\neq m$ such that $h(w_n)=h(w_m)$. Now consider
\[
	v = (c^n,[c\mapsto 0]) = c^n \in T\Gamma.
\]
We have that
\[
	h(\underline{w_n}v) = h(\underline{w_m}v)
\]
but $\underline{w_n}v$ belongs to the direct image~\eqref{eq:c2n} and $\underline{w_m}v$ does not, so $h$ does not recognize the direct image.
\qed

\subsection{The balanced associativity monad}\label{sec:balanced-assoc}

\begin{monad}
Consider a monad $T$ presented by a single binary symbol $\mydot$ and a single equation:

\begin{equation}\label{eq:balanced-assoc}
\begin{tikzpicture}[baseline={(current bounding box.center)}]
  \matrix (m) [matrix of math nodes,row sep=.6em,column sep=.6em,nodes={outer sep=-1.4pt}]
  {  & & \mydot \\
     & \mydot & & x \\
     x & & y\\};
  \draw (m-1-3) -- (m-2-2);
  \draw (m-2-2) -- (m-3-1);
  \draw (m-2-2) -- (m-3-3);
  \draw (m-1-3) -- (m-2-4);
\end{tikzpicture}
\ =\
\begin{tikzpicture}[baseline={(current bounding box.center)}]
  \matrix (m) [matrix of math nodes,row sep=.6em,column sep=.6em,nodes={outer sep=-1.4pt}]
  {  & \mydot \\
     x & & \mydot \\
     & y & & x\\};
  \draw (m-1-2) -- (m-2-1);
  \draw (m-1-2) -- (m-2-3);
  \draw (m-2-3) -- (m-3-2);
  \draw (m-2-3) -- (m-3-4);
\end{tikzpicture}.
\end{equation}
\end{monad}

Elements of $T\Sigma$ do not have a canonical form as simple and intuitive as for the marked words monad, but we shall need only a limited understanding of them. First, notice that for any $a\in \Sigma$, $t\in T\Sigma$ and $n\geq 1$ there is:

\begin{equation}\label{eq:leftize}
\begin{tikzpicture}[baseline={(current bounding box.center)}]
  \matrix (m) [matrix of math nodes,row sep=.6em,column sep=.6em,nodes={outer sep=-1.4pt}]
  {  & \mydot \\
     a & & \mydot \\
     & \mydot & & a\\
     t & & a \\};
  \draw (m-1-2) -- (m-2-1);
  \draw (m-1-2) -- (m-2-3);
  \draw[densely dotted] (m-2-3) -- (m-3-2);
  \draw (m-2-3) -- (m-3-4);
  \draw (m-3-2) -- (m-4-1);
  \draw (m-3-2) -- (m-4-3);
  \draw[pen colour={gray}, decorate, decoration={calligraphic brace, amplitude=4pt}, line width=1.1pt] (1.2,-0.35) -- (0.3,-1.2);
  \draw (1.35,-0.95) node{\scriptsize{$n \text{ times}$}};
\end{tikzpicture}
\ =\
\begin{tikzpicture}[baseline={(current bounding box.center)}]
  \matrix (m) [matrix of math nodes,row sep=.6em,column sep=.6em,nodes={outer sep=-1.4pt}]
  {  & & & \mydot \\
     & & \mydot & & a \\
     & \mydot & & a\\
     a & & t \\};
  \draw[densely dotted] (m-1-4) -- (m-2-3);
  \draw (m-1-4) -- (m-2-5);
  \draw (m-2-3) -- (m-3-2);
  \draw (m-2-3) -- (m-3-4);
  \draw (m-3-2) -- (m-4-1);
  \draw (m-3-2) -- (m-4-3);
  \draw[pen colour={gray}, decorate, decoration={calligraphic brace, amplitude=4pt}, line width=1.1pt] (1.55,0.35) -- (0.65,-0.5);
  \draw (1.7,-0.25) node{\scriptsize{$n \text{ times}$}};
\end{tikzpicture}
\end{equation}

This is proved by simple induction on $n$. The base case $n=1$ is simply~\eqref{eq:balanced-assoc} with $x=a$ and $y=t$, and for the induction step calculate:
\[
\begin{tikzpicture}[baseline={(current bounding box.center)}]
  \matrix (m) [matrix of math nodes,row sep=.6em,column sep=.6em,nodes={outer sep=-1.4pt}]
  {  & & \mydot \\
     & a & & \mydot \\
     & & \mydot & & a\\
     & \mydot & & a \\
     t & & a\\};
  \draw (m-1-3) -- (m-2-2);
  \draw (m-1-3) -- (m-2-4);
  \draw[densely dotted] (m-2-4) -- (m-3-3);
  \draw (m-2-4) -- (m-3-5);
  \draw (m-3-3) -- (m-4-2);
  \draw (m-3-3) -- (m-4-4);
  \draw (m-4-2) -- (m-5-1);
  \draw (m-4-2) -- (m-5-3);
  \draw[pen colour={gray}, decorate, decoration={calligraphic brace, amplitude=4pt}, line width=1.1pt] (1.6,0.05) -- (-0.1,-1.6);
  \draw (1.65,-0.95) node{\scriptsize{$n+1 \text{ times}$}};
\end{tikzpicture}
\quad\stackrel{\eqref{eq:balanced-assoc}}{=}\quad
\begin{tikzpicture}[baseline={(current bounding box.center)}]
  \matrix (m) [matrix of math nodes,row sep=.6em,column sep=.6em,nodes={outer sep=-1.4pt}]
  {  & & \mydot \\
     & \mydot & & a\\
     a & & \mydot \\
     & \mydot & & a\\
     t & & a \\};
  \draw (m-1-3) -- (m-2-2);
  \draw (m-1-3) -- (m-2-4);
  \draw (m-2-2) -- (m-3-1);
  \draw (m-2-2) -- (m-3-3);
  \draw[densely dotted] (m-3-3) -- (m-4-2);
  \draw (m-3-3) -- (m-4-4);
  \draw (m-4-2) -- (m-5-1);
  \draw (m-4-2) -- (m-5-3);
  \draw[pen colour={gray}, decorate, decoration={calligraphic brace, amplitude=4pt}, line width=1.1pt] (1.2,-0.7) -- (0.3,-1.55);
  \draw (1.35,-1.3) node{\scriptsize{$n \text{ times}$}};
\end{tikzpicture}
\stackrel{\text{ind.~ass.}}{=}
\begin{tikzpicture}[baseline={(current bounding box.center)}]
  \matrix (m) [matrix of math nodes,row sep=.6em,column sep=.6em,nodes={outer sep=-1.4pt}]
  {  & & & & \mydot \\
     & & & \mydot & & a \\
     & & \mydot & & a\\
     & \mydot & & a \\
     a & & t\\};
  \draw (m-1-5) -- (m-2-4);
  \draw (m-1-5) -- (m-2-6);
  \draw[densely dotted] (m-2-4) -- (m-3-3);
  \draw (m-2-4) -- (m-3-5);
  \draw (m-3-3) -- (m-4-2);
  \draw (m-3-3) -- (m-4-4);
  \draw (m-4-2) -- (m-5-1);
  \draw (m-4-2) -- (m-5-3);
  \draw[pen colour={gray}, decorate, decoration={calligraphic brace, amplitude=4pt}, line width=1.1pt] (1.9,0.75) -- (.2,-.9);
  \draw (1.95,-0.25) node{\scriptsize{$n+1 \text{ times}$}};
\end{tikzpicture}
\]
Furthermore, for any $a,c\in \Sigma$ and $n\geq 1$ we have:
\begin{equation}\label{eq:kurdebalans}
\begin{tikzpicture}[baseline={(current bounding box.center)}]
  \fill [rounded corners, fill=gray!15] (1,-.3) -- (1.25, -0.05) -- (-0.35,1.6) -- (-1.3,0.6) -- (-1,0.35) -- (-0.3,1.1) -- cycle;
  \fill [rounded corners, fill=gray!15] (1,-1.6) -- (1.25, -1.35) -- (-0.35,.3) -- (-1.3,-0.7) -- (-1,-0.95) -- (-0.3,-.2) -- cycle;
  \matrix (m) [matrix of math nodes,row sep=.6em,column sep=.6em,nodes={outer sep=-1.4pt}]
  {  & \mydot \\
     a & & \mydot \\
     & \mydot & & a\\
     a & & \mydot \\
     & c & & a\\};
  \draw (m-1-2) -- (m-2-1);
  \draw (m-1-2) -- (m-2-3);
  \draw[densely dotted] (m-2-3) -- (m-3-2);
  \draw (m-2-3) -- (m-3-4);
  \draw (m-3-2) -- (m-4-1);
  \draw (m-3-2) -- (m-4-3);
  \draw (m-4-3) -- (m-5-2);
  \draw (m-4-3) -- (m-5-4);
  \draw[pen colour={gray}, decorate, decoration={calligraphic brace, amplitude=4pt}, line width=1.1pt] (-1.3,-0.9) -- (-1.3,1.5);
  \draw (-2,0.3) node{\scriptsize{$n \text{ times}$}};
\end{tikzpicture}
\quad=\quad
\begin{tikzpicture}[baseline={(current bounding box.center)}]
  \matrix (m) [matrix of math nodes,row sep=.6em,column sep=.6em,nodes={outer sep=-1.4pt}]
  {  & & & \mydot \\
     & & \mydot & & a \\
     & \mydot & & a\\
     a & & \mydot \\
     & a & & c\\};
  \draw[densely dotted] (m-1-4) -- (m-2-3);
  \draw (m-1-4) -- (m-2-5);
  \draw (m-2-3) -- (m-3-2);
  \draw (m-2-3) -- (m-3-4);
  \draw[densely dotted] (m-3-2) -- (m-4-3);
  \draw (m-3-2) -- (m-4-1);
  \draw (m-4-3) -- (m-5-2);
  \draw (m-4-3) -- (m-5-4);
  \draw[pen colour={gray}, decorate, decoration={calligraphic brace, amplitude=4pt}, line width=1.1pt] (1.5,0.65) -- (0.6,-.2);
  \draw (1.65,.05) node{\scriptsize{$n \text{ times}$}};
  \draw[pen colour={gray}, decorate, decoration={calligraphic brace, amplitude=4pt}, line width=1.1pt] (-.65,-1.5) -- (-1.5,-.65);
  \draw (-1.7,-1.3) node{\scriptsize{$n \text{ times}$}};
\end{tikzpicture}
\end{equation}
(on the left, the shaded pattern is repeated $n$ times).
This is again proved by induction. The base case $n=1$ is~\eqref{eq:balanced-assoc} with $x=a$ and $y=c$, and for the induction step calculate:
\[
\begin{tikzpicture}[baseline={(current bounding box.center)}]
  \fill [rounded corners, fill=gray!15] (1,.35) -- (1.25, .6) -- (-0.35,2.25) -- (-1.3,1.25) -- (-1,1) -- (-0.3,1.75) -- cycle;
  \fill [rounded corners, fill=gray!15] (1,-0.95) -- (1.25, -0.7) -- (-0.35,.95) -- (-1.3,-.05) -- (-1,-0.3) -- (-0.3,.45) -- cycle;
   \fill [rounded corners, fill=gray!15] (1,-2.25) -- (1.25, -2) -- (-0.35,-.35) -- (-1.3,-1.35) -- (-1,-1.6) -- (-0.3,-.85) -- cycle;
  \matrix (m) [matrix of math nodes,row sep=.6em,column sep=.6em,nodes={outer sep=-1.4pt}]
  {  & \mydot \\
     a & & \mydot \\
     & \mydot & & a \\
     a & & \mydot \\
     & \mydot & & a\\
     a & & \mydot \\
     & c & & a\\};
  \draw (m-1-2) -- (m-2-1);
  \draw (m-1-2) -- (m-2-3);
  \draw (m-2-3) -- (m-3-2);
  \draw (m-2-3) -- (m-3-4);
  \draw (m-3-2) -- (m-4-1);
  \draw (m-3-2) -- (m-4-3);
  \draw[densely dotted] (m-4-3) -- (m-5-2);
  \draw (m-4-3) -- (m-5-4);
  \draw (m-5-2) -- (m-6-1);
  \draw (m-5-2) -- (m-6-3);
  \draw (m-6-3) -- (m-7-2);
  \draw (m-6-3) -- (m-7-4);
  \draw[pen colour={gray}, decorate, decoration={calligraphic brace, amplitude=4pt}, line width=1.1pt] (-1.3,-1.6) -- (-1.3,2.2);
  \draw (-2.3,0.3) node{\scriptsize{$n+1 \text{ times}$}};
\end{tikzpicture}
\stackrel{\text{ind.~ass.}}{=}
\begin{tikzpicture}[baseline={(current bounding box.center)}]
  \matrix (m) [matrix of math nodes,row sep=.6em,column sep=.6em,nodes={outer sep=-1.4pt}]
  {  & & & \mydot \\
     & & a & & \mydot \\
     & & & \mydot & & a \\
     & & \mydot & & a \\
     & \mydot & & a\\
     a & & \mydot \\
     & a & & c\\};
  \draw (m-1-4) -- (m-2-3);
  \draw (m-1-4) -- (m-2-5);
  \draw (m-2-5) -- (m-3-4);
  \draw (m-2-5) -- (m-3-6);
  \draw[densely dotted] (m-3-4) -- (m-4-3);
  \draw (m-3-4) -- (m-4-5);
  \draw (m-4-3) -- (m-5-2);
  \draw (m-4-3) -- (m-5-4);
  \draw[densely dotted] (m-5-2) -- (m-6-3);
  \draw (m-5-2) -- (m-6-1);
  \draw (m-6-3) -- (m-7-2);
  \draw (m-6-3) -- (m-7-4);
  \draw[pen colour={gray}, decorate, decoration={calligraphic brace, amplitude=4pt}, line width=1.1pt] (1.2,0) -- (0.3,-.85);
  \draw (1.4,-.6) node{\scriptsize{$n \text{ times}$}};
  \draw[pen colour={gray}, decorate, decoration={calligraphic brace, amplitude=4pt}, line width=1.1pt] (-1.05,-2.15) -- (-1.9,-1.3);
 \draw (-2.05,-1.95) node{\scriptsize{$n \text{ times}$}};
\end{tikzpicture}
\quad=
\begin{tikzpicture}[baseline={(current bounding box.center)}]
  \matrix (m) [matrix of math nodes,row sep=.6em,column sep=.6em,nodes={outer sep=-1.4pt}]
   { & & & & \mydot \\
     & & & \mydot & & a \\
     & & \mydot & & a \\
     & \mydot & & a\\
     a & & \mydot \\
     & a & & \mydot \\
     & & a & & c\\};
  \draw (m-1-5) -- (m-2-4);
  \draw (m-1-5) -- (m-2-6);
  \draw[densely dotted] (m-2-4) -- (m-3-3);
  \draw (m-2-4) -- (m-3-5);
  \draw (m-3-3) -- (m-4-2);
  \draw (m-3-3) -- (m-4-4);
  \draw (m-4-2) -- (m-5-3);
  \draw (m-4-2) -- (m-5-1);
  \draw (m-5-3) -- (m-6-2);
  \draw[densely dotted] (m-5-3) -- (m-6-4);
  \draw (m-6-4) -- (m-7-3);
  \draw (m-6-4) -- (m-7-5);
  \draw[pen colour={gray}, decorate, decoration={calligraphic brace, amplitude=4pt}, line width=1.1pt] (1.2,0.65) -- (0.3,-.2);
  \draw (1.4,.05) node{\scriptsize{$n \text{ times}$}};
  \draw[pen colour={gray}, decorate, decoration={calligraphic brace, amplitude=4pt}, line width=1.1pt] (-.4,-2.15) -- (-1.25,-1.3);
 \draw (-1.4,-1.95) node{\scriptsize{$n \text{ times}$}};
\end{tikzpicture}
\]
where the last equality follows from~\eqref{eq:leftize} with
\[
\begin{tikzpicture}[baseline={(current bounding box.center)}]
  \matrix (m) [matrix of math nodes,row sep=.6em,column sep=.6em,nodes={outer sep=-1.4pt}]
  {  & \mydot\\
     a & & \mydot \\
     & a & & c\\};
  \draw[densely dotted] (m-1-2) -- (m-2-3);
  \draw (m-1-2) -- (m-2-1);
  \draw (m-2-3) -- (m-3-2);
  \draw (m-2-3) -- (m-3-4);
  \draw[pen colour={gray}, decorate, decoration={calligraphic brace, amplitude=4pt}, line width=1.1pt] (-.35,-.85) -- (-1.2,0);
 \draw (-1.35,-.65) node{\scriptsize{$n \text{ times}$}};
\end{tikzpicture}
=\quad t.
\]

Now, for the alphabet $\Sigma=\{a,b,c\}$, consider the language $L\subseteq T\Sigma$ that consists of all trees of the form
\[
\begin{tikzpicture}[baseline={(current bounding box.center)}]
  \fill [rounded corners, fill=gray!15] (1,-.3) -- (1.25, -0.05) -- (-0.35,1.6) -- (-1.3,0.6) -- (-1,0.35) -- (-0.3,1.1) -- cycle;
  \fill [rounded corners, fill=gray!15] (1,-1.6) -- (1.25, -1.35) -- (-0.35,.3) -- (-1.3,-0.7) -- (-1,-0.95) -- (-0.3,-.2) -- cycle;
  \matrix (m) [matrix of math nodes,row sep=.6em,column sep=.6em,nodes={outer sep=-1.4pt}]
  {  & \mydot \\
     a & & \mydot \\
     & \mydot & & b\\
     a & & \mydot \\
     & c & & b\\};
  \draw (m-1-2) -- (m-2-1);
  \draw (m-1-2) -- (m-2-3);
  \draw[densely dotted] (m-2-3) -- (m-3-2);
  \draw (m-2-3) -- (m-3-4);
  \draw (m-3-2) -- (m-4-1);
  \draw (m-3-2) -- (m-4-3);
  \draw (m-4-3) -- (m-5-2);
  \draw (m-4-3) -- (m-5-4);
  \draw[pen colour={gray}, decorate, decoration={calligraphic brace, amplitude=4pt}, line width=1.1pt] (-1.3,-0.9) -- (-1.3,1.5);
  \draw (-2,0.3) node{\scriptsize{$n \text{ times}$}};
\end{tikzpicture}
\]
for $n\geq 1$. Note that the equation~\eqref{eq:balanced-assoc} does not apply anywhere in such a tree, so every tree of this shape forms a singleton equivalence class with respect to the congruence induced by~\eqref{eq:balanced-assoc}. It is easy to recognize $L$ with an algebra of six values:
\[
	A = \{\Sta,\Stb,\Stc,\StL, \StR,\bot\},
\]
with the operation $\mydot$ interpreted by:
\[
	\Stc \mydot \Stb = \StR
	\qquad
	\Sta \mydot \StR = \StL
	\qquad
	\StL \mydot \Stb = \StR
\]
and $x\mydot y = \bot$ for all other combinations of $x,y\in A$. It is easy to check that equation~\eqref{eq:balanced-assoc} holds in this algebra, indeed
\[
\begin{tikzpicture}[baseline={(current bounding box.center)}]
  \matrix (m) [matrix of math nodes,row sep=.6em,column sep=.6em,nodes={outer sep=-1.4pt}]
  {  & & \mydot \\
     & \mydot & & x \\
     x & & y\\};
  \draw (m-1-3) -- (m-2-2);
  \draw (m-2-2) -- (m-3-1);
  \draw (m-2-2) -- (m-3-3);
  \draw (m-1-3) -- (m-2-4);
\end{tikzpicture}
\ =\
\begin{tikzpicture}[baseline={(current bounding box.center)}]
  \matrix (m) [matrix of math nodes,row sep=.6em,column sep=.6em,nodes={outer sep=-1.4pt}]
  {  & \mydot \\
     x & & \mydot \\
     & y & & x\\};
  \draw (m-1-2) -- (m-2-1);
  \draw (m-1-2) -- (m-2-3);
  \draw (m-2-3) -- (m-3-2);
  \draw (m-2-3) -- (m-3-4);
\end{tikzpicture}
=\quad\bot
\]
for all $x,y\in A$.

Then the unique homomorphism from $h:T\Sigma\to A$ that maps $a$ to $\Sta$, $b$ to $\Stb$ and $c$ to $\Stc$ recognizes $L$:
\[
	t\in L \iff h(t)=\StL.
\]
Now consider an alphabet $\Gamma=\{a,c\}$ and a function $f:\Sigma\to\Gamma$:
\[
	f(a)=f(b)=a, \qquad f(c) = c.
\]
The direct image of $L$ under $Tf$ consists of words as in the equation~\eqref{eq:kurdebalans}.
This language is not recognizable by any finite $T$-algebra. For assume any homomorphism from $T\Gamma$ to some finite algebra $A$, and consider $t_1,t_2,t_3,\ldots\in T\Gamma$ defined by:
\[
\begin{tikzpicture}[baseline={(current bounding box.center)}]
  \matrix (m) [matrix of math nodes,row sep=.6em,column sep=.6em,nodes={outer sep=-1.4pt}]
  {  & \mydot\\
     a & & \mydot \\
     & a & & c\\};
  \draw[densely dotted] (m-1-2) -- (m-2-3);
  \draw (m-1-2) -- (m-2-1);
  \draw (m-2-3) -- (m-3-2);
  \draw (m-2-3) -- (m-3-4);
  \draw[pen colour={gray}, decorate, decoration={calligraphic brace, amplitude=4pt}, line width=1.1pt] (-.35,-.85) -- (-1.2,0);
 \draw (-1.35,-.65) node{\scriptsize{$n \text{ times}$}};
\end{tikzpicture}
=\quad t_n.
\]
Since $A$ is finite, there are some $n\neq m$ such that $h(t_n)=h(t_m)$. Then the terms
\[
\begin{tikzpicture}[baseline={(current bounding box.center)}]
  \matrix (m) [matrix of math nodes,row sep=.6em,column sep=.6em,nodes={outer sep=-1.4pt}]
  {  & & \mydot\\
     & \mydot & & a\\
     t_n & & a\\};
  \draw[densely dotted] (m-1-3) -- (m-2-2);
  \draw (m-1-3) -- (m-2-4);
  \draw (m-2-2) -- (m-3-1);
  \draw (m-2-2) -- (m-3-3);
  \draw[pen colour={gray}, decorate, decoration={calligraphic brace, amplitude=4pt}, line width=1.1pt] (1.2,0) -- (.35,-.85);
 \draw (1.35,-.65) node{\scriptsize{$n \text{ times}$}};
\end{tikzpicture}
\qquad\text{and}\qquad
\begin{tikzpicture}[baseline={(current bounding box.center)}]
  \matrix (m) [matrix of math nodes,row sep=.6em,column sep=.6em,nodes={outer sep=-1.4pt}]
  {  & & \mydot\\
     & \mydot & & a\\
     t_m & & a\\};
  \draw[densely dotted] (m-1-3) -- (m-2-2);
  \draw (m-1-3) -- (m-2-4);
  \draw (m-2-2) -- (m-3-1);
  \draw (m-2-2) -- (m-3-3);
  \draw[pen colour={gray}, decorate, decoration={calligraphic brace, amplitude=4pt}, line width=1.1pt] (1.25,0) -- (.4,-.85);
 \draw (1.4,-.65) node{\scriptsize{$n \text{ times}$}};
\end{tikzpicture}
\]
have the same value under $h$, but the one on the left is in the direct image of $L$ along $Tf$ (due to~\eqref{eq:kurdebalans}) and the one on the right is not, so $h$ does not recognize the direct image.

\subsection{The not-quite-Mal'cev monad}\label{sec:not-quite-malcev}

Recall Monad~\ref{mon:almost-malcev}, presented by a ternary operation symbol $p$ and a unary operation symbol $s$ subject to the equations
\begin{equation}\label{eq:xxy-yxx}
	p(x,x,y) = s(y) = p(y,x,x).
\end{equation}
If, additionally, there were $s(y)=y$, then the monad would immediately have a Mal'cev term and the machinery of Section~\ref{sec:malcev} would apply.\footnote{\ We are grateful to an anonymous reviewer whose insightful question led to a significant improvement of this example.}

Note that any algebra for a monad that has a Mal'cev term induces a $T$-algebra on the same set. For instance, if $(G,\cdot, e)$ is a group, then $(G,p,s)$, where $p(x,y,z):=x\cdot y^{-1}\cdot z$ and $s$ is the identity function, is a $T$-algebra. A finite $T$-algebra that does not have a Mal'cev term is $P=(\{0,1,2,3,4,5,6\},p,s)$ where $p$ is defined as:
\[
    p(x,y,z)=\begin{cases}
		1 & \text{if $x\neq 5$, $y= 5$ and $z\neq 5$}\\
		2 & \text{if $x\neq 6$, $y= 6$ and $z\neq 6$}\\
		3 & \text{if $(x,y,z)=(0,1,2)$,}\\
		4 & \text{otherwise},
        \end{cases}
\]
and with $s(x)=4$ for all $x$. The fact that \eqref{eq:xxy-yxx} hold in $P$ can be easily checked by hand.

Now, consider the alphabet $\Sigma=\{a,b,c\}$, the homomorphism $h:T\Sigma\to P$ such that $h(a)=0$, $h(b)=5$ and $h(c)=6$, and the subset $S=\{3\}\subseteq P$. The recognizable language $L=\overleftarrow{h}(S)$ is the infinite language:
\[
    L=\{p(a,p(x_1,b,x_2),p(x_3,c,x_4))\mid x_i\in T\Sigma, x_1,x_2\neq b, x_3,x_4\neq c\}.
\]

Now consider $\Gamma=\{a,d\}$ and $f:\Sigma\to \Gamma$ defined by $f(b)=f(c)=d$ and $f(a)=a$. We show that that the direct image of $L$ under the surjective letter-to-letter homomorphism $Tf$ is not recognizable. To see this, note that the direct image contains (the equivalence class of) the term $s(a)$, because
\[
	s(a) = p(a,p(d,d,d),p(d,d,d)) = Tg(\ p(a,p(c,b,c),p(b,c,b)\ ).
\]
On the other hand, if $t_1\neq t_2\in T\Gamma$ then the direct image of $L$ does not contain the term
\[
	p(s(t_1),s(t_2),a).
\]
Since $T\Gamma$ is infinite, for any homomorphism $g$ from $T\Gamma$ into a finite $T$-algebra $A$ there are some terms $t_1\neq t_2$ such that $g(t_1)=g(t_2)$. Then the terms
\[
	s(a) = p(s(t_1),s(t_1),a)  \qquad \text{and}\qquad   p(s(t_1),s(t_2),a)
\]
are mapped to the same value by $g$, but only the former one belongs to the direct image of $L$ along $Tf$. 

\section{Related problems}\label{sec:related}

Our main focus has been on recognizable languages being preserved under taking direct images of surjective letter-to-letter homomorphisms. This is because, in the proof of Proposition~\ref{propclosuremso}, which concerns the particular example of the monad of finite words, direct images are taken only along projection homomorphisms that arise from the logical existential quantifier; those homomorphisms are indeed surjective and letter-to-letter. It therefore makes sense to incorporate this particular closure property in an abstract definition of $\mso$ for other monads.

However, it should be noted that Proposition~\ref{propclosuremso} would remain true if surjective letter-to-letter homomorphisms were replaced in its statement by arbitrary letter-to-letter, or indeed by arbitrary homomorphisms. The corresponding closure property would become stronger, so one may expect some monads to become non-examples for these more restrictive variants of monadic $\mso$. In this section we identify some of these monads.

\subsection{Homomorphisms that are not letter-to-letter}

First, let us consider preserving recognizable languages under taking direct images of {\em arbitrary} $T$-algebra homomorphisms, i.e. ones that arise from functions $f:\Sigma\to T\Gamma$ via Kleisli lifting (as opposed to letter-to-letter homomorphims, which are of the form $Tf:T\Sigma\to T\Gamma$ for a map $f:\Sigma\to\Gamma$). This preservation property fails even for very basic monads, such as commutative monoids and seminearrings:

\begin{ex}\label{eximnotrec1}
Recall Monad~\ref{mon:bag}, i.e., the free commutative monoid monad $\Bag$.
The language $a^*\subseteq \Bag\{a,b\}$ is recognized by the commutative monoid $(\{0,1\},\max,0)$
	via the homomorphism $h:\Bag\{a,b\}\to \{0,1\}$ such that $h(a)=0$ and $h(b)=1$. Now, consider the homomorphism $h:\Bag\{a,b\}\to \Bag\{a,b,c\}$ such that
	$h(a)=ab$ and $h(b)=c$. Then, the direct image of $a^*$ is the language $\{w\in \Bag\{a,b,c\}\mid\ w(a)=w(b)\}$ which is not recognizable.\qed
\end{ex}

Another counterexample is:

\begin{monad}\label{mon:seminearring} \label{mon:lastmonad}
A {\em seminearring} is a set equipped with two monoids, i.e., associative binary operations $+$ and $\cdot$ (called {\em horizontal} and {\em vertical} composition; the latter is usually denoted simply by juxtaposition) with units respectively $0$ and $1$, subject to additional axioms
\[
	(x+y)z = xz+yz \qquad \text{and} \qquad 0x = 0.
\]
As expected, the free seminearring monad $T$ is presented by the equational theory of seminearrings.
This is a rather fundamental monad from the perspective of logic: in~\cite[Sec.~6]{B12} it is explained how seminearrings are a reasonable choice for an algebra of trees. Indeed, the free seminearring over a set $\Sigma$ is the set of all unranked forests with holes (called {\em multicontexts} in~\cite{B12}).

\end{monad}

Consider the language
\[
	L = ab^* = \{ab^n\mid n\in\mathbb{N}\} \subseteq T\{a,b\}.
\]
This language is recognized by a finite seminearring (with $5$ elements). Now consider the homomorphism $h:T\{a,b\}\to T\{a,b\}$ with $h(a) = a+a$ and $h(b)=b$. The direct image of $L$ along $h$ is the language
\[
	\{ab^n+ab^n \mid n\in\mathbb{N}\}
\]
which is not recognizable. Indeed, for any seminearring homomorphism $f:T\{a,b\}\to S$ for a finite $S$, there must be some $n\neq m$ such that $f(b^n)=f(b^m)$; then $f(ab^n)=f(ab^m)$ and
\[
	f(ab^n+ab^n) = f(ab^n+ab^m).
\]

\subsection{Letter-to-letter homomorphisms that are not surjective}

One may also consider the preservation of recognizable languages under direct images along arbitrary (i.e. not necessarily surjective) letter-to-letter homomorphisms. In all our counterexamples in Section~\ref{sec:counterexamples}, the failure of preservation seems to be caused by different letters being identified by a homomorphism; in other words, non-injectivity of homomorphisms seems to be the key issue. One may wonder if additionally requiring preservation under non-surjective homomorphisms changes the picture at all. As it turns out, it does.

\begin{ex}
Recall Monad~\ref{mon:reader}, the reader monad from Section~\ref{sec:reader}.
Note that $T\{1\} = \{1^{\omega}\} = \{111111\cdots\}$. Consider the ``full'' language $L=T\{1\}\subseteq T\{1\}$; it is clearly recognizable by the unique homomorphism to the one-element $T$-algebra. However, its direct image along the inclusion $\iota:T\{1\}\to T\{0,1\}$:
\[
	\overrightarrow{T\iota}(L) = \{1^{\omega}\} \subseteq T\{0,1\}
\]
 is not recognizable. Indeed, it is not finitely rectangular according to the terminology of Section~\ref{sec:reader} (and, being a singleton, it is not a union of finitely rectangular sets either), so its recognizability would contradict Lemma~\ref{lem:rectangular2}.
\qed
\end{ex}

\begin{ex}
Recall Monad \ref{mon:fgfgg}, the monad whose finite models are trivial from Section~\ref{sec:fgfgg}. If we consider the inclusion $\iota :\{x\}\to \{x,y\}$ then the direct image of the recognizable language $L=T(\{x\})$ under $T\iota$ is not recognizable. Indeed, the only recognizable languages are $\emptyset$ and $TX$ for any finite $X$ since the only finite model is the trivial model, but $L$ is not one of those languages.
\qed
\end{ex}

\begin{ex}
Recall Monad~\ref{mon:lattice}, the free lattice monad $\FL$ from Section~\ref{sec:freelattice}. We shall consult the first chapter of~\cite{free-lattices} for basic facts about free lattices. First, the following theorem due to Whitman~\cite{whitman} provides an effective procedure for comparing two elements of a free lattice:
\begin{thm}[\cite{free-lattices},~Thm.~1.8]\label{thm:whitman}
For any elements $a,b,c,d\in\FL(\Sigma)$, $a\land b\leq c\lor d$ if and only if:
\[
	a\leq c\lor d, \qquad b\leq c\lor d, \qquad a\land b\leq c, \qquad \text{or} \qquad a\land b\leq d.
\]
\end{thm}
Writing down terms to denote elements of free lattices, one usually implicitly applies the associativity laws and writes e.g. $t_1\land t_2\land t_3$ instead of $(t_1\land t_2)\land t_3$. With this convention in mind, the following canonical form of terms  is considered:
\begin{defi}\label{def:canonical-form}
In a free lattice $\FL(\Sigma)$, a formal join $t = t_1\lor\cdots\lor t_n$ with $n>1$ is in canonical form if:
\begin{itemize}
\item each $t_i$ is in $\Sigma$ or a formal meet in canonical form,
\item $t_i\not\leq t_j$ for all $i\neq j$,
\item if $t_i = \bigwedge t_{ij}$ then $t_{ij}\not\leq t$ for all $j$.
\end{itemize}
Dual conditions define the canonical form of formal meets.
\end{defi}
By~\cite[Thm.~1.17-18]{free-lattices}, every element in $\FL(\Sigma)$ can be presented by a term in canonical form, and moreover this presentation is unique up to commutativity laws.

It will be important to us that Definition~\ref{def:canonical-form} provides a procedure for transforming a term $t$ into canonical form, and that this procedure relies only on removing some redundant subterms from $t$. As a corollary, if an element in $a\in\FL(\Sigma)$ is presented by a term $t$ in canonical form such that some $w\in \Sigma$ appears in $t$, then
$a$ cannot be presented by any term where $w$ does not appear. Indeed, if such a presentation $t'$ existed then the canonical form of $t'$ would not contain $w$ either, which would contradict the uniqueness of canonical presentation.

The last basic fact about free lattices that we shall need is that for $|\Sigma|\geq 3$ the lattice $\FL(\Sigma)$ is infinite and it contains a strictly increasing infinite chain (see~\cite[Ex.~1.24]{free-lattices}).

Now let
\[
	\Sigma = \{p,q,r\} \qquad \Gamma = \{p,q,r,s\}
\]
with $\iota:\Sigma\to\Gamma$ the inclusion function. The language $L=T\Sigma\subseteq T\Sigma$ is recognizable by the unique homomorphism to the one-element lattice. We shall show that its direct image $\overrightarrow{T\iota}(L)\subseteq T\Gamma$ is not recognizable.

To this end, for any pair of terms $t_1<t_2\in \overrightarrow{T\iota}(L)$, consider the term
\[
	t = (t_1\lor s)\land t_2.
\]
Using Theorem~\ref{thm:whitman}, it is easy to check that $s$ is incomparable with both $t_1$ and $t_2$, and that
$t_1 < t < t_2$. Moreover, if $t_1$ and $t_2$ are in canonical form then so is $t$; to prove this use again Theorem~\ref{thm:whitman} and the fact that $s$ is incomparable with $t_1$, $t_2$ and all their subterms. As a result, $t$ cannot be presented by any term where the generator $s$ does not appear; in other words, $t\not\in \overrightarrow{T\iota}(L)$.

We have just proved that between any two distinct but ordered elements of $\overrightarrow{T\iota}(L)$ there is some element not in $\overrightarrow{T\iota}(L)$. Recall that $\FL(\Sigma)$ contains a strictly increasing infinite chain, and since $T\iota$ is an injective function, the set $\overrightarrow{T\iota}(L)$ also contains a strictly increasing infinite chain. Inserting an element out of $\overrightarrow{T\iota}(L)$ between each two neighbouring elements in that chain,
we obtain an infinite, strictly increasing chain that alternates between elements that are in and out of $\overrightarrow{T\iota}(L)$. This implies that the set $\overrightarrow{T\iota}(L)$ is not recognizable.
\end{ex}

\medskip

It should be mentioned that for a weakly cartesian monad $T$, the class $\rec_{T}$ is in fact closed under direct images of all (not necessarily surjective) letter-to-letter homomorphisms. The proof of this is the same as in Section \ref{sec:theproof} but with the full powerset functor $\P$ used instead of $\PP$. Only two points are worth making in this case: (1) the ``projection'' $\pi_2:{\in_X}\to \P X$ is no longer surjective (the empty set is not in the image) and (2) the empty set plays the role of an ``error state'' in a powerset algebra (every letter that is not in the image of the function $g$ is mapped to the empty set and any ``operation'' that involves the empty set has the empty set as a result).

\section{Future work}\label{sec:futurework}

Technically, our main object of study in this paper was a notion: a monad for which languages recognizable by finite algebras are closed under taking direct images along surjective letter-to-letter homomorphisms. For such monads the abstract definition of monadic~\msotext{} makes sense, in that it only describes recognizable languages.
Our various sufficient conditions, examples and counterexamples show that the notion seems rather subtle, but they do not provide a full characterization of it. The search for such a characterization is an obvious direction of further study.

It should also be said that the notion itself is a result of a few design decisions. It could be that its variants, such as the ones described in Section~\ref{sec:related}, lead to simpler characterizations while still covering essentially the same class of practically relevant examples. Some other design decisions we have not even mentioned: for example, is an arbitrary finite algebra the right notion of a recognizing device? Sometimes, in particular for infinitary monads, the answer might not be obvious, and one may want to consider a restricted class of finite algebras instead. (Note that, for an infinitary monad, a finite algebra may not admit a finite description.) Similarly, it is not clear whether one should not restrict even further the class of homomorphisms to take direct images along. This design space deserves careful exploration.

We focussed our attention on monads on the category of sets, but this leaves several natural examples out of scope. For example, the category of ranked sets is a convenient setting to study various algebras of trees and graphs. Other interesting base categories include that of vector spaces, or that of nominal sets. Another interesting idea is to employ the topo-algebraic technology of~\cite{mai1,mai2}, where recognizability of word languages is studied beyond the regular setting. In general, any category where a meaningful notion of direct image of a language can be formulated, may be a territory worth exploring.

The most ambitious goal is to go beyond the ``definable $\Rightarrow$ recognizable'' implication of the correspondence between \msotext{} definability and recognizability of languages. Is there a generic version of the converse implication? How far can Proposition~\ref{propclosuremso} be generalized? What about decidability questions about \msotext{}? The paper~\cite{boj1} makes some initial steps in these and other related directions. However, much remains to be done, all in the spirit of bridging the divide between Structure and Power.

\bibliographystyle{unsrt}
\bibliography{ref}

\end{document}